\documentclass[journal,twocolumn,10pt]{IEEEtran} 

\usepackage[section]{placeins}
\usepackage{graphics} 
\usepackage{epsfig} 
\usepackage{mathptmx} 
\usepackage{amsthm}
\usepackage{amsmath} 
\usepackage{amssymb}  
\usepackage{cuted}
\usepackage{float}
\usepackage{color}
\usepackage{url}
\usepackage{bm}
\usepackage{dblfloatfix}

\usepackage[ruled,vlined,linesnumbered,noresetcount]{algorithm2e} 
\SetKwComment{Comment}{$\triangleright$\ }{}
\usepackage{etoolbox}

\def\Tau{{\mathbf{T}}}

\makeatletter
\patchcmd{\@algocf@start}
  {-1.5em}
  {0pt}
  {}{}
\makeatother

\usepackage{tabularx}

\usepackage{caption,subcaption,cite}
\usepackage[font=small,labelfont=bf]{caption}

\newtheorem{corollary}{Corollary}
\newtheorem{lemma}{Lemma}
\newtheorem{proposition}{Proposition}
\newtheorem{definition}{Definition}
\newtheorem{remark}{Remark}

\newtheorem{condition}{Condition}

\makeatletter
\def\ps@IEEEtitlepagestyle{%
  \def\@oddfoot{\mycopyrightnotice}%
  \def\@evenfoot{}%
}
\def\mycopyrightnotice{%
  \begin{minipage}{\textwidth}
  \centering \scriptsize
  \copyright~2022. This manuscript version is made available under the CC-BY-NC-ND 4.0 license \url{https://creativecommons.org/licenses/by-nc-nd/4.0/}
  \end{minipage}
}
\makeatother

\begin{document}
%
\title{ 
RADAMS: Resilient and Adaptive Alert and Attention Management Strategy against Informational Denial-of-Service (IDoS) Attacks
}
%
%
%

\author{Linan~Huang  and~Quanyan~Zhu
\thanks{L. Huang and Q. Zhu are with the Department of Electrical and Computer Engineering, New York University, Brooklyn, NY, 11201, USA. 
E-mail:\{lh2328,qz494\}@nyu.edu}
\thanks{This work was supported in part by the National Science Foundation (NSF) under Grants ECCS-1847056, CNS-2027884, and BCS-2122060; and in part by Army Research Office (ARO) under Grant W911NF-19-1-0041 and DOE-NE under Grant 20-19829.}
\thanks{
Digital Object Identifier: 10.1016/j.cose.2022.102844. 
A preliminary version of this work \cite{huang2021IDOS} was presented at the $12$-th Conference on Decision and Game Theory for Security 
}
}

\maketitle
\begin{abstract}
Attacks exploiting human attentional vulnerability have posed severe threats to cybersecurity.  
In this work, we identify and formally define a new type of proactive attentional attacks called Informational Denial-of-Service (IDoS) attacks that generate a large volume of feint attacks to overload human operators and hide real attacks among feints. 
We incorporate human factors (e.g., levels of expertise, stress, and efficiency) and empirical \textcolor{black}{psychological} results (e.g., the Yerkes–Dodson law and the sunk cost fallacy) to model the operators' attention dynamics and their decision-making processes along with the real-time alert monitoring and inspection. 
To assist human operators in dismissing the feints and escalating the real attacks timely and accurately, we develop a Resilient and Adaptive Data-driven alert and Attention Management Strategy (RADAMS) that de-emphasizes alerts selectively based on the \textcolor{black}{abstracted category labels of the} alerts. 
RADAMS uses reinforcement learning to achieve a customized and transferable design for various human operators and evolving IDoS attacks. 
The integrated modeling and theoretical analysis lead to the Product Principle of Attention (PPoA), fundamental limits, and the tradeoff among crucial human and economic factors. 
Experimental results corroborate that the proposed strategy outperforms the default strategy and can reduce the IDoS risk by as much as $20\%$. 
Besides, the strategy is resilient to large variations of costs, attack frequencies, and human attention capacities. 
We have recognized interesting phenomena such as attentional risk equivalency, attacker's dilemma, and the half-truth optimal attack strategy. 
\end{abstract}
%

%
\IEEEpeerreviewmaketitle

%
%
%
%


\section{Introduction}
\IEEEPARstart{H}{uman} vulnerability and human-induced security threats have been a long-standing and fast-growing problem for the security of Industrial Control Systems (ICSs). 
According to Verizon \cite{Verizon2021}, $85\%$ data \textcolor{black}{breaches involve} human errors. 
Attentional vulnerability is one of the representative human vulnerabilities. 
Adversaries have exploited human inattention to launch social engineering attacks and phishing attacks toward employees and users. 
According to the report \cite{Tessian}, $29\%$ of employees fall for a phishing scam, and $36\%$ send a misdirected email, owing to lack of attention. 
These attentional attacks are \textit{reactive} as they exploit the existing human attention patterns. 
On the contrary, \textit{proactive} attentional attacks can strategically change the attention pattern of a human operator or a network administrator. 
For example, an attacker can launch feint attacks to trigger a large volume of alerts and overload the human operators so that operators fail to inspect the alert associated with real attacks \cite{WinNT1}.   
We refer to this new type of attacks as the Informational Denial-of-Service (IDoS) attacks, which aim to deplete the limited attention resources of human operators to prevent them from accurate detection and timely defense. 

IDoS attacks bring significant security challenges to ICSs for the following reasons. 
First, alert fatigue has already been a serious problem in the age of infobesity with terabytes of unprocessed data or manipulated information. 
\textcolor{black}{According to the Ponemon Institute research report \cite{PonemonLLC}, organizations spend nearly $21,000$ hours each year analyzing false alarms, which costs organizations an average of $\$1.27$ million per year.} 
IDoS attacks exacerbate the problem by generating feints to intentionally increase the percentage of false-positive alerts. 
Second, IDoS attacks directly target the human operators and security analysts in the Security Operations Center (SOC) \textcolor{black}{that acts as the `central immune system' in ICSs.}  
Third, as ICSs become increasingly complicated and time-critical, the human operators require higher expertise levels to understand the domain information and detect feints\textcolor{black}{\cite{stouffer2011guide} in time to avoid life-threatening failures or huge economic losses. The SOCs in ICSs are usually understaffed, due to these high-standard requirements.}  
\textcolor{black}{Fourth}, since human operators behave differently, and IDoS attacks are a broad class of adaptive attacks, it is challenging (yet highly desirable) to develop a customized and resilient defense. 
\textcolor{black}{Due to the above factors, including the huge economic loss,} 
there is an apparent need to understand this class of proactive attentional attacks, quantify its consequences and risks, and develop associated mitigation strategies. 

\begin{figure*}[]
\centering
 \includegraphics[width=1\linewidth]{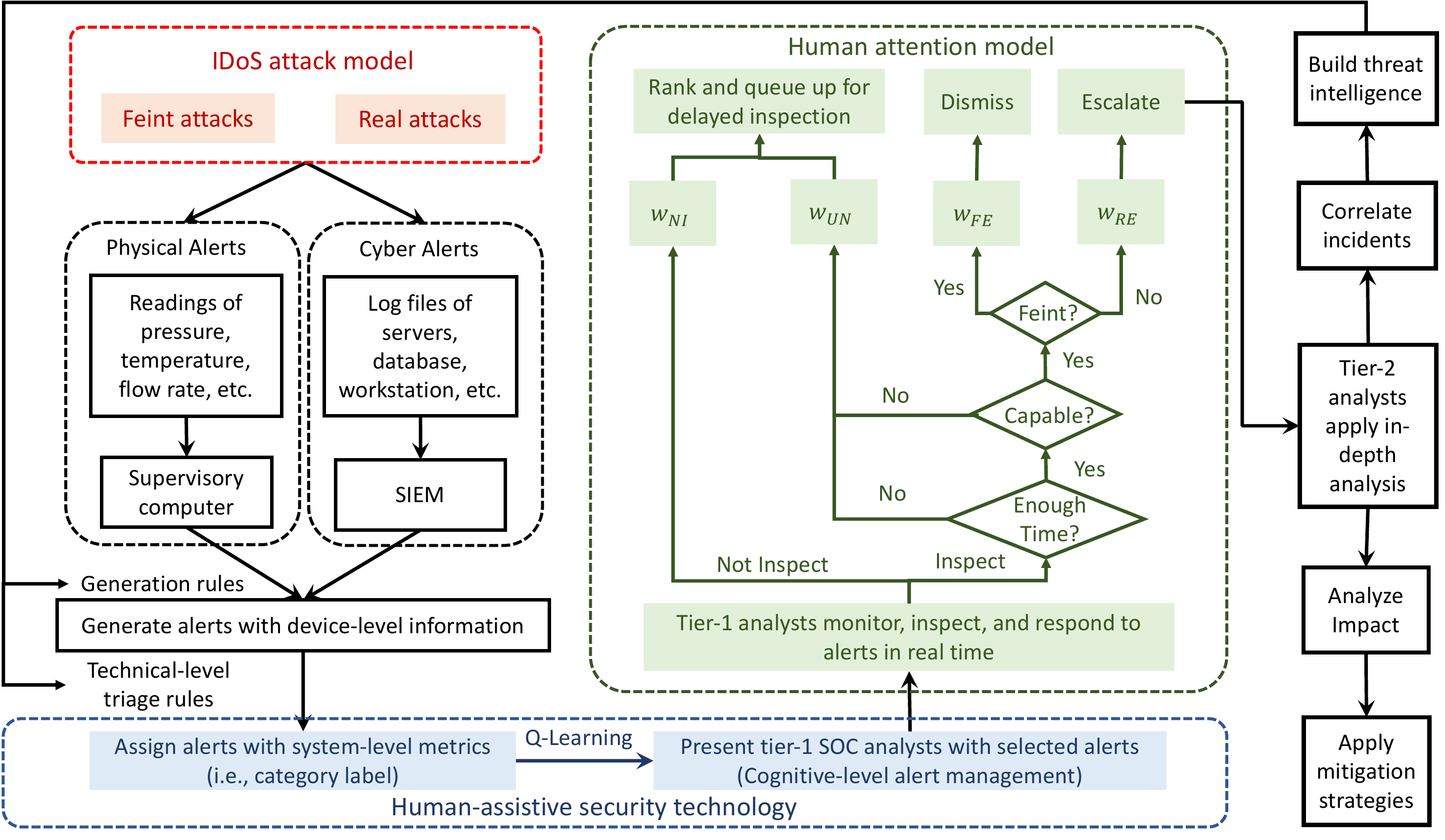}
  \caption{
  The overview diagram of RADAMS against IDoS in ICS, which incorporates the IDoS attack model, the human attention model, and the human-assistive security technology in the red, green, and blue boxes, respectively. 
  \textcolor{black}{RADAMS consolidates the \textit{technical-level} (i.e., generation rules and triage rules in black) and the \textit{cognitive-level} (data-driven human-aware alert  de-emphasis in blue) alert management before the manual inspection in green to reduce the operators' cognitive load.} 
  The modern SOC adopts a hierarchical alert analysis process. 
  The tier-$1$ SOC analysts, also referred to as the operators, are in charge of real-time alert monitoring and inspections. The tier-$2$ SOC analysts are in charge of the in-depth analysis.  
  \textcolor{black}{All processes in black are not the focus of this work.} 
  }
\label{fig:AlertDiagram}
\end{figure*}

To this end, we establish a holistic model of the IDoS attacks, the alert generations, and the human operators' \textcolor{black}{alert responses}. 
In the IDoS attack model, we adopt a Markov renewal process to characterize the sequential arrival of feints and real attacks that target different ICS \textcolor{black}{assets}. 
We \textcolor{black}{define a} \textit{revelation probability} to abstract the alert generation and triage process of existing detection systems. The revelation probability maps the attacks' hidden types and targets stochastically to the associated alerts' observable category labels. 
To model the human operators' attention dynamics and alert responses under the IDoS attacks, we directly incorporate the operators' levels of expertise, stress, and efficiency into the security design based on the existing results from the literature in psychology, including the Yerkes–Dodson law \cite{yerkes1908relation} and the sunk cost fallacy \cite{arkes1985psychology}. 
To assist human operators in alert inspection and response, compensate for their attentional vulnerabilities, and combat IDoS attacks,  we develop human-centered technologies that selectively make some alerts less noticeable based on their category labels. 
Reinforcement learning is applied to make the human-assistive security technology \textit{resilient}, \textit{automatic}, and \textit{adaptive} to various human models and attack scenarios. 

Fig. \ref{fig:AlertDiagram} illustrates the overview diagram of Resilient and Adaptive Alert and Attention Management Strategy (RADAMS). 
\textcolor{black}{We use the following control room scenario to elaborate on the entire process of RADAMS under IDoS attacks. 
Supervisory computers and Security Information and Event Management (SIEM) continuously monitor the physical readings and cyber log files, respectively, to generate alerts with \textit{device-level} information. 
Since manual inspection and response of these alerts (illustrated in green) are indispensable for ICSs at the current stage, RADAMS adopts the following technical-level and cognitive-level automated alert selection schemes, illustrated in black and blue, respectively, to assist manual alert inspection. 
The technical-level alert selection scheme focuses on selecting and prioritizing alerts based on the \textit{device-level} information and abstract \textit{system-level} metrics.
Although the above alert triage process significantly reduces the workload of the manual inspection, a sizeable number of alerts remain to be inspected, especially under a large volume of feints. 
To this end, RADAMS incorporates the cognitive-level alert selection to accommodate the operators' cognition limitation in the subsequent alert inspections. 
After the technical-level and cognitive-level alert management, RADAMS presents the selected alerts to the tier-$1$ SOC analysts in the control room for real-time monitoring and response. 
The alerts associated with the real attack  will be identified and escalated to tier-2 analysts for in-depth analysis. 
The analysis outcomes of tier-$2$ analysts are used to mitigate the current threats and improve the generation rules and technical-level triage rules. 
}

RADAMS enriches the existing alert selection frameworks with the IDoS attack model, the human attention model, and the human-assistive security technology highlighted in red, green, and blue, respectively. 
Through the integrated modeling and theoretical analysis, we obtain the \textit{Product Principle of Attention} (PPoA), which states that the Attentional Deficiency Level (ADL), i.e., the probability of incomplete alert responses, and the risk of IDoS attacks depend on the product of the supply and the demand of human attention resources. 
The closed-form expressions under mild assumptions lead to several fundamental limits, including the minimum ADL and the maximum length of de-emphasized alerts to reduce IDoS risk. 
We explicitly characterize the tradeoff among crucial factors such as the ADL, the reward of alert attention, and the impact of alert inattention. 

Finally, we propose an algorithm to learn the adaptive \textcolor{black}{Attention Management (AM)} strategy based on the operator's alert inspection outcomes. 
We present several case studies based on the simulation of different IDoS attacks and alert inspecting processes. 
The numerical results show that the proposed optimal AM strategy outperforms the default strategy and can effectively reduce the IDoS risk by as much as $20\%$. 
The strategy is also resilient to a large range of cost variations, attack frequencies, and human attention capacities. 
We have observed the phenomenon of \textit{attentional risk equivalency}, which states that the deviation from the optimal to sub-optimal strategies for some category labels can reduce the risk under the default strategy to approximately the same level. 
The results also corroborate that RADAMS can adapt to different category labels to strike a balance of quantity (i.e., inspect more alerts) and quality (i.e., complete alert responses to dismiss feints and escalate real attacks).  
We identify the \textit{attacker's dilemma} where destructive IDoS attacks induce unbearable costs to the attacker. 
We also identify the \textit{half-truth attack strategy} as the optimal IDoS attack strategy when feints are generated at a high cost. 

\subsection{Contribution, Notations, and Organization of the Paper}
\textcolor{black}{
Our main contributions are fourfold. 
First, we have formally defined a new type of attentional attacks called IDoS attacks. 
Second, we propose a consolidated alert and attention management strategy that is explicitly aware of human cognition limitations to defend against IDoS attacks. 
Third, we provide theoretical underpinnings of RADAMS under IDoS attacks and propose a learning algorithm to implement RADAMS in real time. 
Fourth, we present comprehensive case studies to demonstrate the effectiveness, adaptiveness, robustness, and resilience of the proposed assistive strategies.}

The rest of the paper is organized as follows. 
The related work is presented in Section \ref{sec:related works}. 
Sections \ref{sec:IDoS attack model}, \ref{sec:human model}, and \ref{sec:Data-Driven Attention Management} introduce the IDoS attack model, the human operator model, and the human-assistive security technology, respectively. 
\textcolor{black}{We summarize main notations for these three sections in Table \ref{table:notation1}, \ref{table:notation2}, and \ref{table:notation3}, respectively.}  
We analyze the attentional deficiency level and the risk of IDoS attacks in closed form for the class of ambitious operators in Section \ref{sec:Theoretical Analysis}\textcolor{black}{, where the main notations are summarized in Table \ref{table:notation4}}. Section \ref{sec:case study} presents a case study of alert inspection under IDoS attacks and the adaptive AM strategies. Section \ref{sec:conclusion} concludes the paper.

\section{Related Work}
\label{sec:related works}

\subsection{Alert Management}
Previous works have applied various alert management methods during the alert generation, detection, and response processes to mitigate alert fatigue and enhance cybersecurity, \textcolor{black}{as shown in the following three subsections}.  

\subsubsection{Source Management} 
\label{subsubsec:generation}

On the one hand, proactive defense \cite{HUANG2020101660} and deception techniques, including honeypots \cite{huang2019adaptive,huang2020farsighted} and moving target defense \cite{jajodia2011moving}, have managed to \textcolor{black}{reduce alerts at the outset by deterring, delaying, and preventing attacks.} 
On the other hand, previous works have designed incentive mechanisms (e.g., \cite{casey2016compliance,liu2009mitigating}) and information mechanisms (e.g., \cite{DG2021,huang2022zetar}) to enhance insiders' compliance, reduce users' misbehavior, and consequently reduce false positives. 

\subsubsection{Detection Management}
\label{subsubsec:detection}
A rich literature has attempted to develop detection systems capable of reducing false positives while maintaining the ability to detect malicious behaviors.
Methods include statistical analysis \cite{spathoulas2010reducing}, fuzzy inference \cite{elshoush2010reducing}, \textcolor{black}{kernel density estimation \cite{su2019false}, and machine learning approaches \cite{goeschel2016reducing,pietraszek2005data,ohta2008minimizing,bouzar2020rnn}.} 
Alert aggregation and correlation methods \cite{salah2013model} have also been applied to dismiss repeated and innocuous alerts and generate alerts of system-level threat information. 
\textcolor{black}{Recently, the authors in \cite{bryant2020improving} have implemented a hybrid kill-chain based classification model to boost detection rates, improve  alert description, and lower the number of false-positive alerts. 
There is a rich literature on alert filtering and selection, and we refer the readers to \cite{cotroneo2017empirical} for the empirical analysis and validation of these state-of-the-art filtering techniques.} 

\subsubsection{Response Management}
\label{subsubsec:response}
Despite the significant advances in alert reduction methods introduced in Section \ref{subsubsec:generation} and \ref{subsubsec:detection}, the demand for alert inspection still exceeds the operators' capacity. 
To this end, researchers have developed various alert \textcolor{black}{triage} and prioritization approaches \textcolor{black}{that can be classified into the following three categories.}

\textcolor{black}{The first category ranks alerts based on rules. These rules can be generated through fuzzy logic \cite{newcomb2016effective,alsubhi2012fuzmet} and attack graphs \cite{noel2008optimal}. 
Many works have attempted to learn from security experts and automate the process of mining triage rules out of cybersecurity analysts’ operation traces \cite{zhong2016automate,zhong2018learning}. 
The second category assigns scores to alerts and quantitatively optimizes the alert triage process by minimizing the cyber risk.  
The score can be computed through a causal dependency graph of an alert event \cite{hassan2019nodoze}, game-theoretic approaches \cite{laszka2017game}, and the 
Quantitative Value Function (QVF) hierarchy process \cite{TIFS8573836}. 
The authors in \cite{TIFS8573836,ganesan2016dynamic} further incorporate organization-specific factors and constraints into the design of the optimal alert selection. 
The third category relies on data and learning methods. 
Supervised learning \cite{renners2017modeling,bierma2016learning}, deep learning \cite{mcelwee2017deep,aminanto2020threat}, and adversarial reinforcement learning \cite{tong2020finding} are used to prioritize alerts. 
The authors in \cite{zhong2018cyber} have developed a triage operation retrieval system to provide novice analysts with on-the-job suggestions using relevant data triage operations conducted by senior analysts.}

\textcolor{black}{The above three categories of \textit{rule-based}, \textit{risk-aware}, and \textit{data-driven} alert triage methods rank alerts based on their contextual information and organizational factors. Our \textit{human-centered} approach generalizes these classical alert triage approaches by}  explicitly modeling the attentional behaviors of human operators and selecting alerts based on human cognitive capacity. 

\subsection{Feint Attacks and Human Attentional Models}

\textcolor{black}{Feints have been widely studied in sports, military, and biology \cite{Vaccine2017}. 
They are recently used to attack}  
detection systems \cite{CORONA2013201}. 
In particular, the authors in \cite{mutz2003experience, patton2001achilles} have developed tools that can generate false positives by matching detection signatures. 
The tools are tested on SNORT \cite{roesch1999snort}, and the empirical results verify the feasibility of feint attacks on detection systems. 
Compared to \textcolor{black}{these empirical practices of feint attacks that exploit} the vulnerability of detection systems, we focus on the attentional vulnerabilities and the impact of feints on human operators. 
Moreover, we abstract models to formally characterize cyber feint attacks, \textcolor{black}{quantify the risk, and develop human-assistive security technologies}.  

We can classify human vulnerabilities into \textit{acquired} vulnerabilities (e.g., lack of security awareness and noncompliance) and \textit{innate} ones (e.g., bounded attention and rationality) based on whether they can be mitigated through short-term training and security rules. 
Many works (e.g., \cite{wang2021social,DG2021,casey2016compliance}) have emphasized the urgency and necessity to reduce acquired human vulnerability and proposed human-assistive strategies. 
However, few works have focused on mitigation strategies for innate vulnerabilities. 
Visual support systems have been used for rapid cyber event triage \cite{miserendino2017threatvectors} and alert investigations \cite{franklin2017toward}, and eye-tracking data have been incorporated to enhance attention for phishing identification \cite{huang2021advert9819920}. 
The authors in \cite{sundaramurthy2015human} perform an anthropological study in a corporate SOC to model and mitigate security analyst burnout. 
These works lay the foundations of empirical solutions to mitigate human attentional vulnerabilities. Our work combines real-time human behavioral and decision data with the well-identified human factors to enable quantitative characterizations of the empirical relationship such as the Yerkes–Dodson law \cite{yerkes1908relation}. 
The learning-based method for attention management also makes our human-assistive technology adaptive and transferable to various human-technical systems.

\section{IDoS Attacks and Sequential Alert Arrivals}
\label{sec:IDoS attack model}
\begin{table}[h]
\centering
\caption{Summary of Notations in Section \ref{sec:IDoS attack model}
\label{table:notation1}}
\textcolor{black}{
\begin{tabularx}{\columnwidth}{l X} 
     \hline
\textbf{Variable} &  \textbf{Meaning} \\ \hline
$t^k\in [0,\infty)$  &  Arrival time of the $k$-th attack. \\
$\tau^k=t^{k+1} - t^k\in [0,\infty)$       &  Inter-arrival time at attack stage $k\in \mathbb{Z}^{0+}$. \\
$\kappa_{AT}\in \mathcal{K}_{AT}$  & Transition kernel of attacks.\\
$z\in \mathcal{Z}$ &  Probability Density Function (PDF) of the inter-arrival time.\\
$\theta^k\in \Theta:=\{\theta_{FE},\theta_{RE}\}$       &  Attack's type at  attack stage $k\in \mathbb{Z}^{0+}$. \\
$\phi^k\in \Phi$       &  Attack's target at attack stage $k\in \mathbb{Z}^{0+}$.\\
$s^k\in \mathcal{S}$       &  Alert's category label at attack stage $k$.\\
$o(s^k|\theta^k,\phi^k)$ & Revelation kernel of category labels. \\
$b(\theta^k,\phi^k)$ & Steady-state distribution.\\
$\kappa_{CL}\in \mathcal{K}_{CL}$  & Transition kernel of category labels.\\
\hline
\end{tabularx}
}
\end{table}
As illustrated in the first column of Fig. \ref{fig:AlertDiagram}, after the IDoS attacker has generated feint and real attacks, the detection system monitors the readings from physical layers and log files from cyber layers and generates alerts according to the \textit{generation rules}. 
Then, the alerts are sent to the SOC and a triage system automatically generates their category labels (e.g., the alerts' criticality) based on the \textit{\textcolor{black}{technical-level triage} rules}. 
The rules for alert generation and triage are pre-defined and their designs are not the focus of this work. 

\subsection{Feint and Real Attacks of Heterogeneous Targets}
\label{subsec:feint and real}

After \textcolor{black}{the essential preparation stages (e.g., initial intrusion, privilege escalation, and lateral movement), IDoS attacks identify the vulnerable assets as the attack targets and gain control of the ICS to} launch feint and real attacks sequentially, as illustrated by the solid red arrows in Fig. \ref{fig:arrivaldiagram}. 
With a deliberate goal of triggering alerts, feint attacks require fewer resources to craft. 
Although feints have limited impacts on the target system, they aggravate the alert fatigue by depleting human attention resources and preventing human operators from a timely response to real attacks. 
For example, the attacker can attempt to access a database with wrong credentials intentionally, and in the meantime, gradually changes the temperature of the reactor of a nuclear power plant. 
The repeated log-in attempts trigger an excessive number of alerts so that the overloaded human operators fail to pay sustained attention and respond timely to the sensor alerts of the temperature deviation. 

\begin{figure*}[h]
\centering
 \includegraphics[width=1\linewidth]{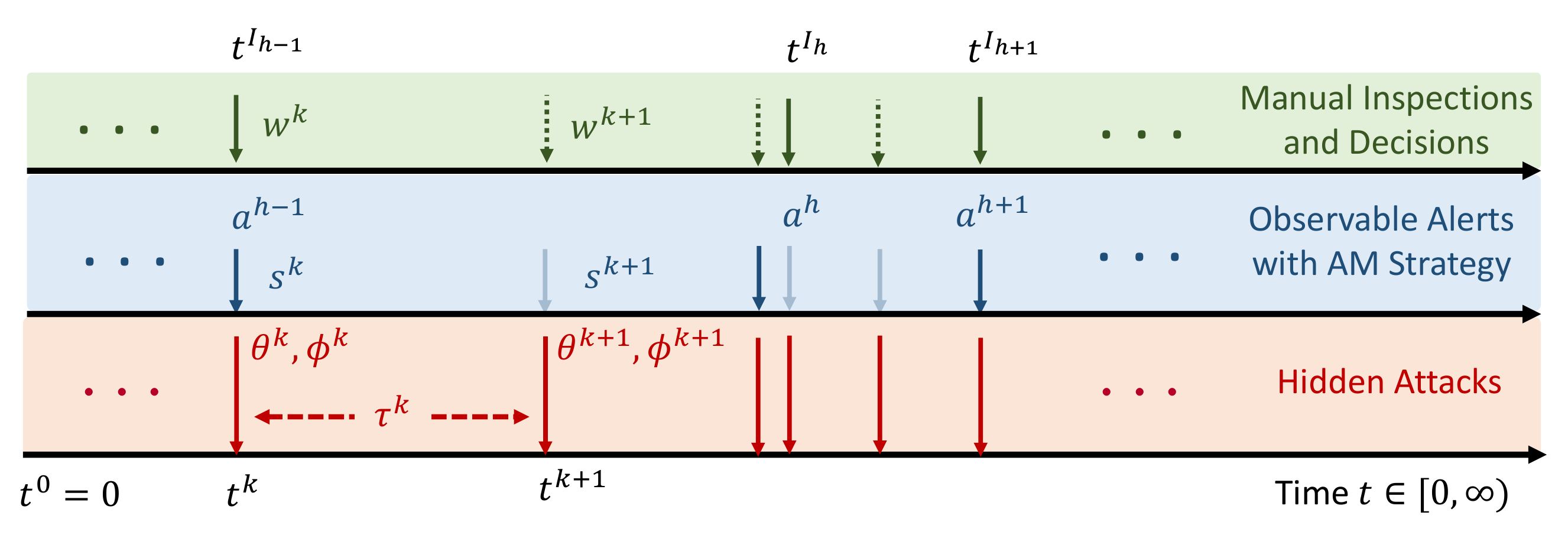}
  \caption{
The timelines of an IDoS attack, alerts under AM strategies, and manual inspections are depicted in red, blue, and green, respectively. 
The inspection stage $h\in \mathbb{Z}^{0+}$ is equivalent to the attack stage $I_h\in \mathbb{Z}^{0+}$. 
The red arrows represent the sequential arrivals of feints and real attacks. 
The semi-transparent blue and the dashed green arrows represent the de-emphasized alerts and the alerts without inspections, respectively. 
  }
\label{fig:arrivaldiagram}
\end{figure*}

We denote feint and real attacks as $\theta_{FE}$ and $\theta_{RE}$, respectively, where $\Theta:=\{\theta_{FE},\theta_{RE}\}$ is the set of attacks' types. 
Each feint or real attack can target cyber \textcolor{black}{assets} (e.g., servers, databases, and workstations) or physical \textcolor{black}{assets} (e.g., sensors of pressure, temperature, and flow rate) in the ICS. 
We define $\Phi$ as the set of the potential attack targets. 
The stochastic arrival of these attacks is modeled as a Markov renewal process where $t^k, k\in \mathbb{Z}^{0+}$, is the time of the $k$-th arrival. 
We refer to the $k$-th attack equivalently as the attack at \textit{attack stage} $k\in \mathbb{Z}^{0+}$ and let $\theta^k\in \Theta$ and $\phi^k\in \Phi$ be the attack's type and target at attack stage $k\in \mathbb{Z}^{0+}$, respectively.  
Define $\kappa_{AT}\in \mathcal{K}_{AT}: \Theta\times \Phi \times \Theta\times \Phi \mapsto [0,1]$ as the transition kernel, where $\kappa_{AT}(\theta^{k+1},\phi^{k+1}|\theta^{k},\phi^{k})$ denotes the probability that the $(k+1)$-th attack has type $\theta^{k+1}\in \Theta$ and target $\phi^{k+1}\in \Phi$ when the $k$-th attack has type $\theta^{k}\in \Theta$ and target $\phi^{k}\in \Phi$. 
The inter-arrival time $\tau^k:=t^{k+1}-t^k$ is a continuous random variable with support $[0,\infty)$ and Probability Density Function (PDF) $z\in \mathcal{Z}:  \Theta\times \Phi \times \Theta\times \Phi \mapsto \mathbb{R}^{0+}$, where  $z(t|\theta^{k+1},\phi^{k+1}, \theta^{k},\phi^{k})$ is the probability that the inter-arrival time is $t$ when the attacks' types and targets at attack stage $k$ and $k+1$ are $\theta^{k},\phi^{k}$ and $\theta^{k+1},\phi^{k+1}$, respectively. 
The values of $\kappa_{AT}\in \mathcal{K}_{AT}$ and $z\in \mathcal{Z}$ are unknown to human operators and the designer of RADAMS. Attackers can adapt $\kappa_{AT}$ and $z$ to different ICSs and alert inspection schemes to achieve the attack goals. 
We formally define IDoS attacks in Definition \ref{def:IDoS}. 
\begin{definition}[\textbf{IDoS Attacks}]
\label{def:IDoS}
An IDoS attack is a sequence of feint and real attacks of heterogeneous targets, which can be characterized by the $4$-tuple $(\Theta, \Phi, \mathcal{K}_{AT}, \mathcal{Z})$. 
\end{definition}

\subsection{\textcolor{black}{Technical-Level} Alert Triage and System-Level Metrics} 
\label{subsec:alert triage}
The alerts triggered by IDoS attacks contain \textit{device-level} contextual information, including the software version, hardware parameters, existing vulnerabilities, and security patches. 
The alert triage process consists of rules that map the device-level information to \textit{system-level} metrics, which helps human operators make timely responses. 
Some essential metrics are listed as follows. 
\begin{itemize}
     \item \textbf{Source} $s_{SO}\in \mathcal{S}_{SO}$:  The ICS sensors or the cyber \textcolor{black}{assets} that the alerts are associated with. 
     \item \textbf{Time Sensitivity} $s_{TS}\in \mathcal{S}_{TS}$: The length of time that the potential attack needs to achieve its attack goals. 
     \item \textbf{Complexity} $s_{CO}\in \mathcal{S}_{CO}$: 
     The degree of effort that a human operator takes to inspect the alert.  
    \item \textbf{Susceptibility} $s_{SU}\in \mathcal{S}_{SU}$:  
    The likelihood that the attack succeeds and inflicts damage on the protected system. 
    \item \textbf{Criticality} $s_{CR}\in \mathcal{S}_{CR}$: 
    The consequence or the impact of the attack's damage. 
\end{itemize}
These alert metrics are observable to the human operators and the RADAMS designer  and form the \textit{category label} of an alert. 
We define the category label associated with the $k$-th alert as $s^k:=(s_{SO}^k,s_{TS}^k,s_{CO}^k,s_{SU}^k,s_{CR}^k) \in \mathcal{S}$, where $\mathcal{S}:=\mathcal{S}_{SO}\times  \mathcal{S}_{TS}\times \mathcal{S}_{CO}\times \mathcal{S}_{SU}\times  \mathcal{S}_{CR}$. 
The joint set $\mathcal{S}$ can be adapted to suit the organization's needs in the security practice. For example, we have  $\mathcal{S}_{TS}=\emptyset$ if time sensitivity is unavailable or unimportant. 


The \textcolor{black}{technical-level} alert triage process establishes a stochastic connection between the hidden types and targets of the IDoS attacks and the observable category labels of the associated alerts. 
Let $o(s^k|\theta^k,\phi^k)$ be the probability of obtaining category label $s^k\in \mathcal{S}$, when the associated attack has type $\theta^k\in \Theta$ and target $\phi^k\in \Phi$. 
The revelation kernel $o$ reflects the quality of the alert triage. 
For example, feints with lightweight resource consumption usually have a limited impact. Thus, a high-quality triage process should classify the associated alert as low criticality with a high probability. 
Letting $b(\theta^k,\phi^k)$ denote the probability that the $k$-th attack has type $\theta^k$ and target $\phi^k$ at the steady-state, we can compute the steady-state distribution $b$ in closed form based on $\kappa_{AT}$.  
Then, the transition of category labels at different attack stages is also Markov and is represented by $\kappa_{CL}\in \mathcal{K}_{CL}:\mathcal{S}\times\mathcal{S}\mapsto [0,1]$. 
We can compute $\kappa_{CL}= \frac{ \Pr(s^{k+1},s^k) }{ \sum_{s^{k+1}\in \mathcal{S}} \Pr(s^{k+1},s^k) }$  based on $\kappa_{AT}, o, b$, 
where $\Pr(s^{k+1},s^k) =  \sum_{\theta^k,\theta^{k+1}\in \Theta} \sum_{\phi^k,\phi^{k+1}\in \Phi}
   \kappa_{AT}(\theta^{k+1},\phi^{k+1}|\theta^k,\phi^k)  \allowbreak 
   o(s^k|\theta^k,\phi^k)  \allowbreak
   o(s^{k+1}|\theta^{k+1},\phi^{k+1}) \allowbreak
   b(\theta^k,\phi^k) 
   $. 
In this work, we focus on the case where the detection system introduces the same delay between attacks and their triggered alerts. 
Since the sequences of attacks and alerts have a one-to-one mapping, we can consider zero delay time without loss of generality. Hence, the sequence of alerts associated with an IDoS attack  $(\Theta, \Phi,\mathcal{K}_{AT}, \mathcal{Z})$ is also a Markov renewal process characterized by the $3$-tuple $(\mathcal{S},  \mathcal{K}_{CL},\mathcal{Z})$. 


\section{Human Attention Model under IDoS Attacks}
\label{sec:human model}

\begin{table}[h]
\centering
\caption{Summary of Notations in Section \ref{sec:human model}
\label{table:notation2}}
\textcolor{black}{
\begin{tabularx}{\columnwidth}{l X} 
     \hline
\textbf{Variable} &  \textbf{Meaning} \\ \hline
$w_{FE},w_{RE},w_{UN},w_{NI}$    & Alert dismissal, alert escalation, uninspected alerts, and inadequate alert response.  \\
$w^k\in \mathcal{W}$  &  Operator's alert response at attack stage $k$.\\
$\kappa^{\Delta k}_{SW}(s^{k+\Delta k}|s^k)$ & Operator's default switching probability. \\
$D_{max} (s^k)\in \mathbb{R}^+$ & Maximum Allowable Delay (MAD) for responding to alerts of category label $s^{k}\in \mathcal{S}$.\\ 
$t_{AoI}^k=t-t^k$ & $k$-th alert's Age of Information (AoI).\\
$y_{EL}\in \mathcal{Y}_{EL}$ & Operator's expertise level. \\
$\bar{d}(y_{EL},s^k, \theta^k, \phi^k)\in \mathbb{R}^+$ & Average inspection time to reach a complete alert response $w_{FE}$ or $w_{RE}$.\\
${d}(y_{EL},s^k, \theta^k, \phi^k)$ & Actual Inspection Time Needed (AITN).\\
$n^t\in \mathbb{Z}^{0+}$ & Number of alerts that arrive during the current inspection up to time $t\in [0,\infty)$. \\ 
$y_{SL}^t=f_{SL}(n^t)\in \mathbb{R}^+$ & Operator's stress level at time $t$. \\
$\omega^t=f_{LOE}(y_{SL}^t)\in [0,1]$ & Operator's Level of Operational Efficiency (LOE) at time $t$. \\
$\bar{n}(y_{EL},s^k)\in \mathbb{R}^{0+}$  & Attention threshold.\\
$\tilde{\omega}^{t_1,t_2}:=\int_{t_1}^{t_2} \omega^t d{t}$ & Effective Inspection Time (EIT) during inspection time $[t_1,t_2]$. \\
$p_{SP}(y_{EL},s^k,\theta^k, \phi^k)$ & Probability of a complete response.\\ 
\hline
\end{tabularx}
}
\end{table}
An SOC typically adopts a hierarchical alert analysis \cite{zimmerman2014cybersecurity}. 
The attention model in this section applies to the tier-$1$ SOC analysts, or the operators, who are in charge of monitoring, inspecting, and responding to alerts in real time. 
As illustrated by the green box in Fig. \ref{fig:AlertDiagram}, the operators choose to inspect certain alerts, dismiss the feints, and escalate the real attacks to tier-$2$ SOC analysts for in-depth analysis. 
The in-depth analysis can last hours to months, during which the tier-$2$ analysts correlate incidents from different \textcolor{black}{assets} in the ICS over long periods to build threat intelligence and analyze the impact. 
The threat intelligence is then incorporated to form and update the generation rules of the detection system and \textcolor{black}{triage} rules of the triage process. 


\subsection{Alert Responses}
\label{subsec:Alert Responses}
Due to the high volume of alerts and the potential short-term surge arrivals, human operators cannot inspect all alerts in real time. The uninspected alerts receive an alert response $w_{NI}$. 
Whether the operator chooses to inspect an alert depends on the switching probability in Section \ref{subsec:Real-Time Monitoring}. 

When the operator inspects an alert, he can be distracted by the arrival of new alerts and switch to newly-arrived alerts without completing the current inspection. 
We elaborate on the attention dynamics in Section \ref{subsec:Attentional Factors}. 
The alert with incomplete inspection is labeled by $w_{UN}$. 
Besides the insufficient inspection time, the operator's cognitive capacity constraint can also prevent him from determining whether the alert is triggered by a feint or a real attack. 
In this work, we consider prudent operators. When they cannot determine the attack's type after a full inspection, the associated alert is labeled as $w_{UN}$, \textcolor{black}{as shown in the green flowchart of Fig. \ref{fig:AlertDiagram}}. 
We elaborate on how the insufficient inspection time and the operator's cognitive capacity constraint lead to $w_{UN}$, i.e., referred to as the \textit{inadequate alert response}, in Section \ref{subsec:inspection and decision-making}. 
The alerts labeled as $w_{NI}$ and $w_{UN}$ are ranked and queued up for delayed inspections at later stages. 


When the operator successfully completes the alert inspection with a deterministic decision, he either dismisses the alert (denoted by $w_{FE}$) or escalates the alert to tier-$2$ SOC analysts for in-depth analysis (denoted by $w_{RE}$), as shown in Fig. \ref{fig:AlertDiagram}. 
We use $w^k\in \mathcal{W}:=\{w_{FE},w_{RE},w_{UN},w_{NI}\}$ to denote the operator's response to the alert at attack stage $k\in \mathbb{Z}^{0+}$. 
We can extend the set $\mathcal{W}$ to suit the organization's security practice. For example, some organizations let the operators report their estimations and confidence levels concerning incomplete alert inspection, i.e., divide the label $w_{UN}$ into finer subcategories.  
Then at later stages, the delayed inspection can prioritize the alerts based on the estimations and confidence levels. 

\subsection{Probabilistic Switches within Allowable Delay}
\label{subsec:Real-Time Monitoring}
Alerts are monitored in real time when they arrive.
When the category label of the new alert indicates higher time sensitivity, susceptibility, or criticality, the operator can delay the current inspection (i.e., label the alert under inspection as $w_{UN}$) and switch to inspect the new alert. 
We denote $\kappa^{\Delta k}_{SW}(s^{k+\Delta k}|s^k)$ as the operator's \textcolor{black}{default} switching probability when the previous alert at attack stage $k$ and the new alert at stage $k+\Delta k, \Delta k\in \mathbb{Z}^{+}$, have category label $s^k\in \mathcal{S}$ and $s^{k+\Delta k}\in \mathcal{S}$, respectively. 
As a probability measure, 
\begin{equation}
\label{eq:prob measure}
    \sum_{\Delta k=1}^{\infty} \sum_{s^{k+\Delta k} \in \mathcal{S}} \kappa^{\Delta k}_{SW}(s^{k+\Delta k}|s^k) \equiv 1 , \forall k\in \mathbb{Z}^{0+}, \forall s^k\in \mathcal{S}. 
\end{equation}
Since the operator cannot observe the attack's hidden type and hidden target, the switching probability $\kappa_{SW}^{\Delta k}$ is independent of $\theta^k,\phi^k$ and $\theta^{k+1},\phi^{k+1}$. 
The switching probability depends on the time that the operator has already spent on the current inspection. 
For example, an operator becomes less likely to switch after spending a long time inspecting an alert of low criticality or beyond his capacity, which can lead to the Sunk Cost Fallacy (SCF).  

We denote \textcolor{black}{$D_{max}(s^{k})\in \mathbb{R}^{+}$} as the Maximum Allowable Delay (MAD) \textcolor{black}{for alerts of category label $s^{k}\in \mathcal{S}$}.  
At time $t\geq t^k$, the $k$-th alert's Age of Information (AoI) \cite{Yates2021AgeOI} is defined as $t^k_{AoI}:=t-t^k$. 
This work focuses on time-critical ICSs where a defensive response for the $k$-th alert of  \textcolor{black}{category label $s^{k}\in \mathcal{S}$} is only effective if the alert's AoI is within the MAD, i.e., \textcolor{black}{$t^k_{AoI}\leq D_{max}(s^{k})$}.  
Therefore, the operator will be reminded when an alert's AoI exceeds the MAD so that he can switch to monitor and inspect new alerts. 
The MAD and the reminder scheme help mitigate the SCF when the operators are occupied with old alerts and miss the chance to monitor and inspect new alerts in real time.

\subsection{Attentional Factors}
\label{subsec:Attentional Factors}
We identify the following human and environmental factors affecting operators' alert inspection and response processes. 
\begin{itemize}
    \item The operator's expertise level denoted by $y_{EL}\in \mathcal{Y}_{EL}$. 
    \item The $k$-th alert's category label $s^k\in \mathcal{S}$. 
    \item The $k$-th attack's type $\theta^k$ and target $\phi^k$. 
    \item The operator's stress level $y_{SL}^t\in \mathbb{R}^+$, which changes with time $t$ as new alerts arrive. 
\end{itemize}

The first three factors are the static attributes of the analyst, the alert, and the IDoS attack, respectively. 
They determine the average inspection time, denoted by $\bar{d}(y_{EL},s^k, \theta^k, \phi^k)\in \mathbb{R}^+$, to reach a \textit{complete response} $w_{FE}$ or $w_{RE}$. 
For example, if the inspected alert is of low complexity, the operator can reach a complete response in a shorter time. Also, it takes a senior operator less time on average to reach a complete alert response than a junior one does.  
We use $d(y_{EL},s^k,\theta^k, \phi^k)$ to represent the Actual Inspection Time Needed (AITN) when the operator is of expertise level $y_{EL}$, the alert is of category label $s^k$, and the attack has type $\theta^k$ and target $\phi^k$. AITN $d(y_{EL},s^k,\theta^k, \phi^k)$ is a random variable with mean $\bar{d}(y_{EL},s^k,\theta^k, \phi^k)$. 

The fourth factor reflects the temporal aspect of human attention during the inspection process. 
Evidence has shown that the continuous arrival of the alerts can increase the stress level of human operators \cite{ancker2017effects}, and $52\%$ of employees attribute their mistakes to stress \cite{Tessian}. 
We denote $n^t\in \mathbb{Z}^{0+}$ as the number of alerts that arrives during the current inspection up to time $t\in [0,\infty)$ and model the operator's stress level $y_{SL}^t$ as an increasing function $f_{SL}$ of $n^t$, i.e., $y_{SL}^t=f_{SL}(n^t)$. 
At time $t\in [0,\infty)$, the human operator's Level of Operational Efficiency (LOE), denoted by $\omega^t\in [0,1]$, is a function $f_{LOE}$ of the stress level $y_{SL}^t$, i.e., 
\begin{equation}
    \omega^t=f_{LOE}(y_{SL}^t)=(f_{LOE}\circ f_{SL})(n^t), \forall t\in [0,\infty). 
\end{equation}
Based on the Yerkes–Dodson law, the function $f_{LOE}$ follows an inverse $U$-shape that contains the following two regions.  
In region one, a small number of alerts result in a moderate stress level and allow human operators to inspect the alert efficiently. 
In region two, the LOE starts to decrease when the number of alerts to inspect is beyond some threshold \textcolor{black}{$\bar{n}(y_{EL},s^k)\in \mathbb{R}^{0+}$,} and the human operator is overloaded. 
The value of the \textit{attention threshold} $\bar{n}(y_{EL},s^k)$ depends on the operator's expertise level $y_{EL}\in \mathcal{Y}_{EL}$ and the alert's category label $s^k\in \mathcal{S}$. 
For example, it requires more (resp. fewer) alerts (i.e., higher (resp. lower) attention threshold) to overload a senior (resp. an inexperienced) operator. 
We can also adapt the value of $\bar{n}(y_{EL},s^k)$ to different scenarios. In the extreme case where all alerts are of high complexity and create a heavy cognitive load, we let $\bar{n}(y_{EL},s^k)=0, \forall y_{EL}\in \mathcal{Y}_{EL},s^k\in \mathcal{S}$, and the LOE decreases monotonously with the number of alert arrivals during an inspection. 



\subsection{Alert Responses under Time and Capacity Limitations}
\label{subsec:inspection and decision-making}
After we identify attentioinal factors in Section \ref{subsec:Attentional Factors}, we illustrate their impacts on the operators' alert responses as follows.  
We define the Effective Inspection Time (EIT) during inspection time $[t_1,t_2]$ as the integration  $\tilde{\omega}^{t_1,t_2}:=\int_{t_1}^{t_2} \omega^t d{t}$. 
When the operator is overloaded and has a low LOE during $[t_1,t_2]$, the EIT $\tilde{\omega}^{t_1,t_2}$ is much shorter than the actual inspection time $t_2-t_1$. 

Suppose that the operator of expertise level $y_{EL}$ inspects the $k$-th alert for a duration of $[t_1,t_2]$. 
If the EIT has exceed the AITN $d(y_{EL},s^k,\theta^k, \phi^k)$, then the operator can reach a complete response $w_{FE}$ or $w_{RE}$ with a high success probability denoted by $p_{SP}(y_{EL},s^k,\theta^k, \phi^k)\in [0,1]$. 
However, when $\tilde{\omega}^{t_1,t_2}<d(y_{EL},s^k,\theta^k, \phi^k)$, it indicates that the operator has not completed the inspection, and the alert response concerning the $k$-th alert is $w^k=w_{UN}$. 
The success probability $p_{SP}$ depends on the operator's capacity to identify attacks' types, which leads to the definition of the capacity gap below. 
\begin{definition}[\textbf{Capacity Gap}]
\label{def:capacity gap}
For an operator of expertise level $y_{EL}\in \mathcal{Y}_{EL}$, we define $p_{CG}(y_{EL},s^k,\theta^k, \phi^k):=1-p_{SP}(y_{EL},s^k,\theta^k, \phi^k)$ as his capacity gap  to inspect an alert with category label $s^k\in \mathcal{S}$, type $\theta^k\in \Theta$, and target $\phi^k\in \Phi$ defined in Section \ref{sec:IDoS attack model}. 
\end{definition}

\section{Human-Assistive Security Technology for \textcolor{black}{Cognitive-Level Alert Management}}
\label{sec:Data-Driven Attention Management}

\begin{table}[h]
\centering
\caption{Summary of Notations in Section \ref{sec:Data-Driven Attention Management}
\label{table:notation3}}
\textcolor{black}{
\begin{tabularx}{\columnwidth}{l X} 
     \hline
\textbf{Variable} &  \textbf{Meaning} \\ \hline
$I_h\in \mathbb{Z}^{0+}$, $t^{I_h}\in [0,\infty)$ & Index and time of the alert under the $h$-th inspection (i.e., inspection stage $h\in \mathbb{Z}^{0+}$). \\
$a_m\in \mathcal{A}$ & Attention management (AM) strategy of period $m\in \mathbb{Z}^+$. \\
$a^h\in \mathcal{A}$ & AM action at inspection stage $h\in \mathbb{Z}^{0+}$. \\
$\bar{\kappa}^{I_{h+1}-I_h, a^h}_{SW}(s^{I_{h+1}}|s^{I_h})$ &  Operator's switching probability under $a^h$. \\
$\bar{c}(w^{k},s^{k})\in \mathbb{R}$ & Stage cost. \\
${c}(s^{I_h},a^{h})\in \mathbb{R}$ & Expected Consolidated Cost (ECoC). \\
$\tilde{c}(s^{I_h},a^{h})\in \mathbb{R}$ &  Consolidated Cost (CoC). \\
$\sigma^0,\sigma^*\in \Sigma$  & Default and optimal AM strategy. \\
\hline
\end{tabularx}
}
\end{table}
As illustrated in Section \ref{sec:human model}, the frequent arrival of alerts triggered by IDoS attacks can overload the human operator and reduce the LOE and the EIT. 
To compensate for the human's attentional limitation, we can intentionally make some alerts less noticeable, e.g., without sounds or in a light color, based on their category labels. 
As illustrated by the blue box in Fig. \ref{fig:AlertDiagram}, based on the category labels from the \textcolor{black}{technical-level} triage process, RADAMS automatically emphasizes and de-emphasizes alerts, \textcolor{black}{referred to as the cognitive-level alert management}, and then presents them to the tier 1 SOC analysts. 

\subsection{Adaptive Attention Management Strategy}
\label{subsec:Adaptive AM Strategy}
In this work, we focus on the class of \textcolor{black}{AM} strategies, denoted by $\mathcal{A}:=\{a_m\}_{m\in \{0,1,\cdots,M\}}$, that de-emphasize consecutive alerts. 
As explained in Section \ref{subsec:Alert Responses}, the operator can only inspect some alerts in real time. 
Thus, we use $I_h\in \mathbb{Z}^{0+}$ and $t^{I_h}\in [0,\infty)$ to denote the index and the time of the alert under the $h$-th inspection; i.e., the inspection stage $h\in \mathbb{Z}^{0+}$ is equivalent to the attack stage $I_h\in \mathbb{Z}^{0+}$. 
Whenever the operator starts a new inspection at inspection stage $h\in \mathbb{Z}^{0+}$, RADAMS determines the AM action $a^h\in \mathcal{A}$ for the $h$-th inspection based on the stationary strategy $\sigma\in \Sigma: \mathcal{S} \mapsto \mathcal{A}$ that is adaptive to the category label of the $h$-th alert. 
We illustrate the timeline of the manual inspections and the AM strategies in green and blue, respectively, in Fig. \ref{fig:arrivaldiagram}. 
The solid and dashed green arrows indicate the inspected and uninspected alerts, respectively. 
The non-transparent and semi-transparent blue arrows indicate the emphasized and de-emphasized alerts, respectively. 
At inspection stage $h$, if $a^h=a_m$, RADAMS will make the next $m$ alerts less noticeable; i.e., the alerts at attack stages $I_h+1,\cdots,I_h+m$ are de-emphasized. 
Denote $\bar{\kappa}^{I_{h+1}-I_h, a^h}_{SW}(s^{I_{h+1}}|s^{I_h})$ as the operator's switching probability to these de-emphasized alerts under the AM action $a^h\in \mathcal{A}$. 
    Analogously to \eqref{eq:prob measure}, the following holds for all $h\in \mathbb{Z}^{0+}$ and $a^h\in \mathcal{A}$, i.e.,  
    \begin{equation}
        \label{eq:prob measureAM}
        \sum_{I_{h+1}=I_h+1}^{\infty} \sum_{s^{I_{h+1}} \in \mathcal{S}} \bar{\kappa}^{I_{h+1}-I_h, a^h}_{SW}(s^{I_{h+1}}|s^{I_h}) \equiv 1 , \forall s^{I_h}\in \mathcal{S}. 
    \end{equation}
The deliberate de-emphasis on selective alerts brings the following tradeoff. 
On the one hand, these alerts do not increase the operator's stress level, and the operator can pay sustained attention to the alert under inspection with high LOE and EIT. 
On the other hand, these alerts do not draw the operator's attention, and the operator is less likely to switch to them during the real-time monitoring and inspections.

Since the operator may switch to inspect a de-emphasized alert with switching probability $\bar{\kappa}^{I_{h+1}-I_h, a^h}_{SW}$ (e.g., the $h$-inspection in Fig. \ref{fig:arrivaldiagram}), RADAMS recomputes the AM strategy and implements the new strategy whenever the operator has started to inspect a new alert. 
Although the operator can switch unpredictably, Proposition \ref{prop:Markov} shows that the transition of the inspected alerts' category labels is Markov. 

\begin{proposition}
\label{prop:Markov}
For a stationary AM strategy $\sigma\in \Sigma$, the set of random variables $(\mathbf{S}^{\mathbf{I}_h},\mathbf{T}^{\mathbf{I}_h})_{h\in \mathbb{Z}^{0+}}$ is a Markov renewal process.  
\end{proposition}

\begin{proof}
The sketch of the proof includes two steps. 
First, we prove that the state transition from $s^{I_h}$ to $s^{I_{h+1}}$ is Markov for all $h\in\mathbb{Z}^{0+}$. 
Due to the uncertainty of switching in inspection, the transition stage $\mathbf{I}_{h+1}$ is also a random variable for all $h\in \mathbb{Z}^{0+}$, and we can represent the transition probability as
\begin{equation*}
        \resizebox{.99\hsize}{!}{$
         \Pr(\mathbf{S}^{\mathbf{I}_{h+1}}=s^{I_{h+1}}|s^{I_{h}}) = \sum_{l=1}^{\infty}  \Pr(\mathbf{I}_{h+1}=I_h+l) \cdot \Pr(\mathbf{S}^{\mathbf{I}_{h+1}}=s^{I_{h+1}}|s^{I_{h}}), 
            $}
\end{equation*}
where $\Pr(\mathbf{I}_{h+1}=I_h+l)$ is the probability that the $(h+1)$-th inspection happens at attack stage $I_h+l$.  
The term $\Pr(\mathbf{S}^{\mathbf{I}_{h+1}}=s^{I_{h+1}}|s^{I_{h}})$  is Markov and can be computed based on $\kappa_{CL}$. 
The term $\Pr(\mathbf{I}_{h+1}=I_h+l)$ depends on $d(y_{EL},s^{I_h+l'},\theta^{I_h+l'}, \phi^{I_h+l'})$,  $\kappa_{SW}^{l'}$, $\bar{\kappa}_{SW}^{l'}$, $\tau^{l'}$, for all $l'\in \{1,\cdots,l\}$.  
Since $s^{I_h+l'},\theta^{I_h+l'},\phi^{I_h+l'}, l'\in \{1,\cdots,l\}$, are all stochastically related to $s^{I_h}$ and $s^{I_{h+1}}$ based on $o$, $\kappa_{AT}$ and $\kappa_{CL}$, the term  $\Pr(\mathbf{I}_{h+1}=I_h+l)$ depends on $s^{I_h}$ and $s^{I_{h+1}}$ for all $l\in \mathbb{Z}^+$. 

Then, we show that the distribution of the inter-arrival time $\Tau_{IN}^{\mathbf{I}_h,m}:=\mathbf{T}^{\mathbf{I}_{h+1}}-\mathbf{T}^{{I}_h}$ only depends on $s^{I_h}$ and $s^{I_{h+1}}$. 
Analogously, the cumulative distribution function of $\Tau_{IN}^{\mathbf{I}_h,m}$ is
\begin{equation*}
    \begin{split}
        \Pr(\Tau_{IN}^{\mathbf{I}_h,m}\leq t)=\sum_{l=1}^{\infty}  \Pr(\mathbf{I}_{h+1}=I_h+l) \cdot \Pr(\Tau_{IN}^{I_h,m}\leq t),  
    \end{split}
\end{equation*}
and hence we arrive at the Markov property. 
\end{proof}

\subsection{Stage Cost and Expected Cumulative Cost}
\label{subsec:Stage Cost and Expected Cumulative Cost}
For each alert at attack stage $k\in \mathbb{Z}^{0+}$, RADAMS assigns a stage cost \textcolor{black}{$\bar{c}(w^{k},s^{k})\in \mathbb{R}$} \textcolor{black}{to evaluate the outcomes of alert response $w^k\in \mathcal{W}$ under the category label $s^{k}\in \mathcal{S}$. 
The value of the cost varies under different scenarios. In this work, we can estimate it using} the salary of SOC analysts and the estimated loss of the associated attack. 
For example, $\bar{c}(w_{UN},s^{I_h})$ and $\bar{c}(w_{NI},s^{I_h})$ are positive costs as 
those alerts without a complete response incur additional workloads. The delayed inspections also expose the organization to the threats of time-sensitive attacks. 
On the other hand, $\bar{c}(w_{FE},s^{I_h})$ and $\bar{c}(w_{RE},s^{I_h})$ are negative costs because the alerts with complete alert response $w_{FE}$ and $w_{RE}$ reduce the workload of tier 2 SOC analysts and enable them to obtain threat intelligence. 

When the operator starts a new inspection at inspection stage $h+1$, RADAMS will evaluate the effectiveness of the AM strategy for the $h$-th inspection. 
The performance evaluation is reflected by the Expected Consolidated Cost (ECoC) ${c}: \mathcal{S}\times \mathcal{A} \mapsto \mathbb{R}$ at each inspection stage $h\in \mathbb{Z}^{0+}$. 
We denote the realization of ${c}(s^{I_h},a^h)$ as the Consolidated Cost (CoC) $\Tilde{c}^{I_h}(s^{I_h}, a^h)$. 
Since the AM strategy $\sigma$ at each inspection stage can affect the future human inspection process and the alert responses, we define the Expected Cumulative Cost (ECuC) $u(s^{I_h}, \sigma):=\sum_{h=0}^{\infty} \gamma^h c(s^{I_h},\sigma(s^{I_h}))$ under adaptive strategy $\sigma\in \Sigma$ as the long-term performance measure. 
The goal of the assistive technology is to design the optimal adaptive strategy $\sigma^*\in \Sigma$ that minimizes the ECuC $u$ under the presented IDoS attack based on the category label $s^{I_h}\in \mathcal{S}$ at each inspection stage $h$.  
We define $v^*(s^{I_h}):=\min_{\sigma\in \Sigma} u(s^{I_h}, \sigma)$ as the optimal ECuC when the category label is $s^{I_h}\in \mathcal{S}$.  
We refer to the \textit{default AM strategy} $\sigma^0\in \Sigma$ as the one when no AM action is applied under all category labels, i.e., $\sigma^0(s^{I_h})=a_0, \forall s^{I_h}\in \mathcal{S}$.  


\subsection{Reinforcement Learning}
\label{subsec:Reinforcement Learning}
Due to the absence of the following exact model parameters, RADAMS has to learn the optimal AM strategy $\sigma^*\in \Sigma$ based on the operator's alert responses in real time. 
\begin{itemize}
    \item Parameters of the IDoS attack model (e.g., $\kappa_{AT}$ and $z$) and the alert generation model (e.g., $o$) in Section \ref{sec:IDoS attack model}. 
    \item Parameters of the human attention model (e.g., $f_{LOE}$ and $f_{Sl}$),  inspection model (e.g., $\kappa_{SW}^{\Delta k}$, $\bar{\kappa}^{I_{h+1}-I_h, a^h}_{SW}$, and $d$), and alert response model (e.g., $y_{EL}$ and $p_{SP}$) in Section \ref{sec:human model}. 
\end{itemize}

Define $Q^h(s^{I_h},a^h)$ as the estimated ECuC during the $h$-th inspection when the category label is $s^{I_h}\in \mathcal{S}$ and the AM action is $a^h$.  
Based on Proposition \ref{prop:Markov}, the state transition is Markov, which enables Q-learning as follows. 
\begin{equation}
\label{eq:Qlearning}
\begin{split}
& Q^{h+1}(s^{I_{h}},a^h):=(1-\alpha^h(s^{I_h},a^h))Q^{h}(s^{I_h},a^h) \\
& \quad 
+ \alpha^h(s^{I_h},a^h)[\Tilde{c}^{I_h}(s^{I_h}, a^h)
+\gamma \min_{a'\in \mathcal{A}} Q^h({s}^{I_{h+1}},a')], 
\end{split}
\end{equation}
where $s^{I_h}$ and  $s^{I_{h+1}}$ are the observed category labels of the alerts at the attack stage $I_h$ and $I_{h+1}$, respectively. 
When the learning rate  $\alpha^h(s^{I_h},a^h)\in (0,1)$ satisfies $\sum_{h=0}^\infty \alpha^h(s^{I_h},a^h)=\infty, \sum_{h=0}^\infty (\alpha^h(s^{I_h},a^h))^2<\infty, \forall s^{I_h}\in \mathcal{S}, \forall a^h\in \mathcal{A}$, and all state-action pairs are explored infinitely, $\min_{a'\in \mathcal{A}} \allowbreak
Q^h(s^{I_h},a')$ converges to the optimal ECuC $v^*(s^{I_h})$ with probability $1$ as $h\rightarrow \infty$. 
At each inspection stage $h\in \mathbb{Z}^{0+}$, RADAMS selects AM strategy $a^h\in \mathcal{A}$ based on the $\epsilon$-greedy policy; i.e., RADAMS chooses a random action with a small probability $\epsilon\in [0,1]$, and the optimal action $arg\min_{a'\in \mathcal{A}} Q^h(s^{I_h},a')$ with probability $1-\epsilon$. 

We present the algorithm to learn the adaptive AM strategy based on the operator's real-time alert monitoring and inspection process in Algorithm \ref{algorithm:human simulation}. 
\begin{algorithm}
\small
 \caption{ \label{algorithm:human simulation} Algorithm to Learn the Adaptive AM strategy based on the Operator's Real-Time Alert Inspection}
\SetAlgoLined
\textbf{Input} $K$: The total number of attack stages\; 
 \textbf{Initialize} 
 The operator starts the $h$-th inspection under AM action $a^h\in \mathcal{A}$; $I_h=k_0$; $\Tilde{c}^{I_h}(s^{I_h}, a^h)=0$\;
   \For{$k  \leftarrow k_0+1$ \KwTo  $K$}  
   {
    \eIf{The operator has finished the $I_h$-th alert (i.e., $\text{EIT}>\text{AITN}$),} 
    {
    \eIf{Capable (i.e., $\textrm{rand}\leq p_{SP}(y_{EL},s^k,\theta^k, \phi^k)$)}{
     Dismiss (i.e., $w^{I_h}=w_{FE}$) or escalate (i.e., $w^{I_h}=w_{RE}$) the $I_h$-th alert\;}
     {
     Queue up the $I_h$-th alert, i.e., $w^{I_h}=w_{UN}$\;}
     $\Tilde{c}^{I_h}(s^{I_h}, a^h)= \Tilde{c}^{I_h}(s^{I_h}, a^h) + \bar{c}(w^{I_h},s^{I_h})$\;
     $I_{h+1}  \leftarrow  k$; The operator starts to inspect the $k$-th alert with category label $s^{I_{h+1}}$\; 
        Update $Q^{h+1}(s^{I_{h}},a^h)$ via \eqref{eq:Qlearning} and obtain the AM action  $a^{h+1}$ by $\epsilon$-greedy policy\; 
    $\Tilde{c}^{h+1}(s^{I_{h+1}},a^{h+1})=0$;  $h \leftarrow h+1$\;
    }
    {
    \eIf{The operator chooses to switch \textbf{or} The MAD is reached, i.e.,  \textcolor{black}{$t^k-t^{I_h}\geq D_{max}(s^{I_h})$}
    }
    {
    Queue up the $I_h$-th alert (i.e., $w^{I_h}=w_{UN}$)\;
    $\Tilde{c}^{I_h}(s^{I_h}, a^h)= \Tilde{c}^{I_h}(s^{I_h}, a^h) + \bar{c}(w_{UN},s^{I_h})$\;
   $I_{h+1}  \leftarrow  k$; The operator starts to inspect the $k$-th alert with category label $s^{I_{h+1}}$\; 
        Update $Q^{h+1}(s^{I_{h}},a^h)$ via \eqref{eq:Qlearning} and obtain the AM action  $a^{h+1}$ by $\epsilon$-greedy policy\; 
    $\Tilde{c}^{h+1}(s^{I_{h+1}},a^{h+1})=0$;  $h \leftarrow h+1$\;
    }
    {The operator continues the inspection of the $I_h$-th alert with decreased LOE\;
    The $k$-th alert is queued up for delayed inspection (i.e., $w^k=w_{NI}$)\; 
     $\Tilde{c}^{I_h}(s^{I_h}, a^h)=\Tilde{c}^{I_h}(s^{I_h}, a^h)+\bar{c}(w_{NI},s^{k})$\; 
    }
    }
  }
 \textbf{Return}  $Q^h(s,a), \forall s\in \mathcal{S},a\in \mathcal{A}$\; 
\end{algorithm}
Each simulation run corresponds to the operator's work shift of $24$ hours at the SOC. 
Since the SOC can receive over $10$ thousand of alerts in each work shift, we can use infinite horizon to approximate the total number of attack stages $K>10,000$. 
Whenever the operator starts to inspect a new alert at inspection stage $I_{h+1}$, RADAMS applies Q-learning in \eqref{eq:Qlearning} based on the category label $s^{I_{h+1}}$ of the newly arrived alert and determines the AM action $a^{h+1}$ for the $h+1$ inspection based on the $\epsilon$-greedy policy as shown in lines $12$ and $19$ of Algorithm \ref{algorithm:human simulation}.  
The CoC $\Tilde{c}^{I_h}(s^{I_h}, a^h)$ of the $h$-th inspection under the AM action $a^h\in\mathcal{A}$ and the category label $s^{I_h}$ of the inspected alert can be computed iteratively based on the stage cost $\bar{c}(w^k,s^k)$ of the alerts during the attack stage $k\in \{I_h,\cdots,I_{h+1}-1\}$, as shown in lines $13$, $20$, and $24$ of Algorithm \ref{algorithm:human simulation}.

\section{Theoretical Analysis}
\label{sec:Theoretical Analysis}

\begin{table}[h]
\centering
\caption{Summary of Notations in Section \ref{sec:Theoretical Analysis}
\label{table:notation4}}
\textcolor{black}{
\begin{tabularx}{\columnwidth}{l X} 
     \hline
\textbf{Variable} &  \textbf{Meaning} \\ \hline
$p_{UN}(s^{I_h},a^h)$ &  Attentional Deficiency Level (ADL). \\
$\beta>0$ & Poisson arrival rate.\\
$\bar{z}$ & PDF of Erlang distribution with shape $m+1$ and rate $\beta$.\\
$p^h_{SD}(w^{I_h} | s^{I_h},a^h;\theta^{I_h},\phi^{I_h})$ & Probability that the operator makes alert response $w^{I_h}$ at inspection stage $h$. \\
$\lambda(s^{I_h},m,\phi^{I_h})$ & Expected reward of a complete alert response.\\ 
\hline
\end{tabularx}
}
\end{table}
In Section \ref{sec:Theoretical Analysis}, we focus on the class of ambitious operators who attempt to inspect all alerts, i.e., $\kappa_{SW}(s^{k+\Delta k}|s^k)= \mathbf{1}_{\{\Delta k= 1\}}, \forall s^k,s^{k+\Delta k}\in\mathcal{S}, \forall \Delta k\in\mathbb{Z}^{+}$. 
To assist this class of operators, the implemented AM action $a_m, m\in \{0,1,\cdots,M\}$, chooses to make the selected alerts fully unnoticeable. 
Then, under $a_m\in \mathcal{A}$, the operator at inspection stage $h$ can pay sustained attention to inspect the alert of category label $s^{I_h}\in \mathcal{S}$ for $m+1$ attack stages. 
Moreover, the operator switches to the new alert at attack stage $I_{h+1}$, i.e., 
$\sum_{s^{I_{h}+m+1} \in \mathcal{S}} \bar{\kappa}^{I_{h+1}-I_h, a_m}_{SW}(s^{I_{h}+m+1}|s^{I_h})=\mathbf{1}_{\{I_{h+1}-I_h=m+1 \}}$. 
Throughout the section, we omit the variable of the expertise level $y_{EL}$ in functions $d,\bar{d},p_{SP}$, and $p_{CG}$ because $y_{EL}$ is a constant for all attack stages. 

\subsection{Security Metrics}
\label{subsec:Security Metrics}
We propose two security metrics in Definition \ref{def:securitymetric} to evaluate the performance of ambitious operators under IDoS attacks and different AM strategies. 
The first metric, denoted as $p_{UN}(s^{I_h},a^h)$, is the probability that the operator chooses $w_{UN}$ during the $h$-th inspection under the category label $s^{I_h}\in\mathcal{S}$ and AM action $a^h\in \mathcal{A}$. 
This metric reflects the Attentional Deficiency Level (ADL) of the IDoS attack. 
For example, as the attackers generate more feints at a higher frequency, the operator is persistently distracted by the new alerts, and it becomes unlikely for him to fully respond to an alert. The ADL $p_{UN}(s^{I_h},a^h)$ is high in this scenario.  
We use the ECuC $u(s^{I_h}, \sigma)$ as the second metric that evaluates the \textit{IDoS risk} under the category label $s^{I_h}\in\mathcal{S}$ and the AM strategy $\sigma\in \Sigma$. 
For both metrics, smaller values are preferred. 

\begin{definition}[\textbf{Attentional Deficiency Level and Risk}]
\label{def:securitymetric}
Under category label $s^{I_h}\in \mathcal{S}$ and the stationary AM strategy $\sigma\in \Sigma$, we define $p_{UN}(s^{I_h},\sigma(s^{I_h}))$ and $u(s^{I_h}, \sigma)$ as the Attentional Deficiency Level (ADL) and the risk of the IDoS attacks defined in Section \ref{sec:IDoS attack model}, respectively. 
\end{definition}

\subsection{Closed-Form Computations} 
\label{subsec:Closed-form computation}
The Markov renewal process that characterizes the IDoS attack or the associated alert sequence follows a Poisson process when Condition \ref{condition:poissonProcess} holds. 
\begin{condition}[\textbf{Poisson Arrival}]
\label{condition:poissonProcess}
The inter-arrival \textcolor{black}{times $\tau^k, \forall k\in \mathbb{Z}^{0+}$,   are independent and} exponentially distributed random variables with the same arrival rate denoted by $\beta>0$, i.e., $z(\tau |\theta^{k+1},\phi^{k+1},\theta^{k},\phi^{k})= \beta e^{{-\beta \tau}}, \tau\in [0,\infty)$ for all $\theta^{k+1},\theta^{k}\in \Theta$ and  $\phi^{k+1},\phi^{k}\in \Phi$. 
\end{condition}

Recall that random variable $\mathbf{T}_{IN}^{I_h,m}$ represents the inspection time of the $I_h$-th alert under the AM action $a^h=a_m\in \mathcal{A}$. 
For the ambitious operators under AM action $a_m\in \mathcal{A}$ at inspection stage $h$, the next inspection happens at attack stage $I_{h+1}=I_h+m+1$. Thus, $I_{h+1}$ is no longer a random variable. 
As a summation of $m+1$ i.i.d. exponential distributed random variables of rate $\beta$, $\mathbf{T}_{IN}^{I_h,m}$ follows an \textit{Erlang distribution} denoted by \textcolor{black}{PDF function} $\bar{z}$ with shape $m+1$ and and rate $\beta>0$ when condition \ref{condition:poissonProcess} holds, i.e., $ \bar{z}(\tau)=\frac {\beta^{m+1} \tau^{m}e^{-\beta \tau}}{m!}, \tau\in [0,\infty)$. 

Denote $p^h_{SD}(w^{I_h} | s^{I_h},a^h;\theta^{I_h},\phi^{I_h})$ as the probability that the operator makes alert response $w^{I_h}$ at inspection stage $h$. 
To obtain a theoretical underpinning, we consider the case where the AITN equals the average inspection time, i.e., $d(s^k,\theta^k, \phi^k)=\bar{d}(s^k,\theta^k, \phi^k)$.  
Then, the operator under AM action $a_m$ makes a complete alert response (i.e., $w^{I_h}\in \{w_{FE},w_{RE}\}$) at inspection stage $h$ for category label $s^{I_h}$ if the inspection time $\tau_{IN}^{I_h,m}$ is greater than the AITN. 
The probability of the above event can be represented as 
$\int_{d(s^{I_h},\theta^{I_h}, \phi^{I_h})}^{\infty} p_{SP}(s^{I_h},\theta^{I_h}, \phi^{I_h}) \bar{z}(\tau) d\tau =p_{SP}(s^{I_h},\theta^{I_h}, \phi^{I_h})    \cdot  \sum_{n=0}^{m} \frac{1}{n!} \allowbreak e^{-\beta d(s^{I_h},\theta^{I_h}, \phi^{I_h}) }  (\beta d(s^{I_h},\theta^{I_h}, \phi^{I_h}) )^n$, which leads to 
\begin{equation}
\label{eq:probofUN}
   \begin{split}
        & p_{SD}^h (w_{UN} | s^{I_h},a_m;{\theta}^{I_h}, \phi^{I_h})
        =1- p_{SP}(s^{I_h},\theta^{I_h}, \phi^{I_h})   \\
        &  \quad\quad \quad
        \cdot \sum_{n=0}^{m} \frac{1}{n!} e^{-\beta d(s^{I_h},\theta^{I_h}, \phi^{I_h}) } (\beta d(s^{I_h},\theta^{I_h}, \phi^{I_h}) )^n. 
   \end{split}
\end{equation}
Then, the ADL $ p_{UN}(s^{I_h},a^h)$ can be computed as 
\begin{equation}
\label{eq:severityLevelcloseform}
\begin{split}
       \sum_{\theta^{I_h}\in \Theta,\phi^{I_h}\in \Phi} \Pr({\theta}^{I_h}, \phi^{I_h}|s^{I_h})    \cdot p_{SD}^h(w_{UN} | s^{I_h},a^h;{\theta}^{I_h}, \phi^{I_h}) , 
\end{split}
\end{equation}
where the conditional probability $\Pr({\theta}^{I_h}, \phi^{I_h}|s^{I_h})$ can be computed via the Bayesian rule, i.e., 
$\Pr({\theta}^{I_h}, \phi^{I_h}|s^{I_h})= \frac{ o(s^{I_h}|{\theta}^{I_h}, \phi^{I_h}) b({\theta}^{I_h}, \phi^{I_h}) }{ \sum_{\theta^{I_h}\in \Theta,\phi^{I_h}\in \Phi} o(s^{I_h}|{\theta}^{I_h}, \phi^{I_h}) b({\theta}^{I_h}, \phi^{I_h}) }$.

We can compute the ECoC $c(s^{I_h},a_m)$ explicitly as 
\begin{equation}
\label{eq:consolidated cost}
\begin{split}
        &    c(s^{I_h},a_m)=m \bar{c}(w_{NI},s^{I_h})    + \sum_{\theta^{I_h}\in \Theta,\phi^{I_h}\in \Phi} \Pr({\theta}^{I_h}, \phi^{I_h}|s^{I_h}) \\
        & \quad \quad \quad \quad \quad 
        \cdot \sum_{w^{I_h}\in \mathcal{W}}  p_{SD}^h(w^{I_h}|s^{I_h},a_m;\theta^{I_h},\phi^{I_h})  \bar{c}(w^{I_h},s^{I_h}).  
\end{split}
\end{equation}
For prudent operators in Section \ref{subsec:Alert Responses}, we have 
\begin{equation}
\label{eq:alwayscorrect}
    \begin{split}
    \resizebox{0.88 \hsize}{!}{$
        p_{SD}^h(w_{i} | s^{I_h},a^h;\theta_{i},\phi^{I_h})=1-p_{SD}^h(w_{UN} | s^{I_h},a^h;\theta_{i},\phi^{I_h}),
        $}
    \end{split}
\end{equation} 
for all  $i\in \{FE,RE\}, s^{I_h}\in \mathcal{S}, a^h\in \mathcal{A}, \phi^{I_h}\in \Phi, h\in \mathbb{Z}^{0+}$. 
Plugging \eqref{eq:alwayscorrect} into \eqref{eq:consolidated cost}, we can simplify the ECoC $c(s^{I_h},a_m)$ as 
\begin{equation}
     \label{eq:SimpliedConCost}
            \resizebox{0.91 \hsize}{!}{$
     \begin{split}
         & c(s^{I_h},a_m) = \sum_{\phi^{I_h}\in \Phi}  \sum_{i\in \{FE,RE\}}   \Pr(\theta_{i}, \phi^{I_h}|s^{I_h})  \cdot p_{SD}^h(w_{i} | s^{I_h},a_m;\theta_{i},\phi^{I_h}) \\
        & \quad \quad\quad \quad \cdot [\bar{c}(w_{i},s^{I_h})- \bar{c}(w_{UN},s^{I_h}) ] +  m\bar{c}(w_{NI},s^{I_h})+\bar{c}(w_{UN},s^{I_h}). 
     \end{split}
             $}
\end{equation}
As shown in Proposition \ref{prop:exp_product}, the ADL and the risk are monotone function of $\beta d(s^{I_h},\theta^{I_h}, \phi^{I_h})$ for each AM strategy.  
\begin{proposition}
\label{prop:exp_product}
If condition \ref{condition:poissonProcess} holds, then the ADL $p_{UN}(s^{I_h},\sigma(s^{I_h}))$ and the risk $u(s^{I_h},\sigma)$ of an IDoS attack under category label $s^{I_h}\in \mathcal{S}$ and AM strategy $\sigma\in\Sigma$ increase in the value of the product $\beta d(s^{I_h},\theta^{I_h}, \phi^{I_h})$. 
\end{proposition}

\begin{proof}
First, since $p_{SD}^h(w_{UN})$ in \eqref{eq:probofUN} increases monotonously with respect to the product $\beta d(s^{I_h},\theta^{I_h}, \phi^{I_h})$, the values of $p_{SD}^h(w_{FE})$ and $p_{SD}^h(w_{RE})$ in \eqref{eq:alwayscorrect} decrease monotonously with respect to the product. 
Plugging \eqref{eq:probofUN} into \eqref{eq:severityLevelcloseform}, we obtain that $p_{UN}(s^{I_h},a_m)$ in \eqref{eq:severityLevelFinalform} under any $a_m\in \mathcal{A}$ and $s^{I_h}\in \mathcal{S}$ is a summation of functions increasing in $\beta d(s^{I_h},\theta^{I_h}, \phi^{I_h})$. 
\begin{equation}
\label{eq:severityLevelFinalform}
\begin{split}
    & p_{UN}(s^{I_h},a_m) =    \sum_{\phi^{I_h}\in \Phi} \sum_{i\in \{FE,RE\}} 
    \Pr({\theta}_{i}, \phi^{I_h}|s^{I_h}) [ 1 - \\
        &  
         p_{SP}(s^{I_h},\theta_{i}, \phi^{I_h})  
        \cdot \sum_{n=0}^{m} \frac{1}{n!} e^{-\beta d(s^{I_h},\theta_{i}, \phi^{I_h}) } (\beta d(s^{I_h},\theta_{i}, \phi^{I_h}) )^n ]. 
\end{split}
\end{equation}

Second, since $\bar{c}(w_{FE},s^{I_h})$ and $\bar{c}(w_{RE},s^{I_h})$ are negative, and $\bar{c}(w_{UN},s^{I_h})$ is positive, the ECoC in \eqref{eq:SimpliedConCost} decreases with $\beta d(s^{I_h},\theta^{I_h}, \phi^{I_h})$ under any $a_m\in \mathcal{A}$ and $s^{I_h}\in \mathcal{S}$. 
Then, the risk also decreases with the product, due to the monotonicity of the Bellman operator \cite{bertsekas1996neuro}.  
\end{proof}

\begin{remark}[\textbf{Product Principle of Attention (PPoA)}] 
 On the one hand, as $\beta$ increases, the feint and real attacks arrive at a higher frequency on average, resulting in a higher demand of attention resources from the human operator. 
 On the other hand, as $d(s^{I_h},\theta^{I_h}, \phi^{I_h})$ increases, the human operator requires a longer inspection time to determine the attack's type, leading to a lower supply of attention resources. 
 \textcolor{black}{Proposition \ref{prop:exp_product} characterizes the PPoA; i.e., for any stationary AM strategy $\sigma\in \Sigma$, the ADL and the risk of IDoS attacks depend on the product of the supply and demand of attention resources.} 
\end{remark}

\subsection{Fundamental Limits under AM strategies}
\label{subsec:Fundamental limits under AM strategies}
Section \ref{subsec:Fundamental limits under AM strategies} aims to show the fundamental limits of the IDoS attack's ADL, the ECoC, and the risk under different AM strategies. 
Define the shorthand notation: 
$\underline{p}(s^{I_h}):=\sum_{\phi^{I_h}\in \Phi} \sum_{i\in \{FE,RE\}} \Pr({\theta}_{i}, \phi^{I_h}|s^{I_h}) {p}_{CG}(s^{I_h},\theta_{i}, \phi^{I_h}) $. 

\begin{lemma}
\label{lemma:SLwithAM}
If Condition \ref{condition:poissonProcess} holds and $M\rightarrow \infty$, then for each $s^{I_h}\in \mathcal{S}$, the ADL $p_{UN}(s^{I_h},a_m)$ decreases strictly to $\underline{p}(s^{I_h})    $ as $m$ increases.  
\end{lemma}
\begin{proof}
Since $\frac{1}{n!} e^{-\beta d(s^{I_h},\theta^{I_h}, \phi^{I_h})}) (\beta d(s^{I_h},\theta^{I_h}, \phi^{I_h}) )^n>0$ for all $ m\in \{0,\cdots, M\}$, the value of $p_{UN}(s^{I_h},a_m)$ in \eqref{eq:severityLevelFinalform} strictly decreases as $m$ increases. 
Moreover, since $\lim_{m\rightarrow \infty} \sum_{n=0}^{m} \frac{1}{n!} e^{-\beta d(s^{I_h},\theta^{I_h}, \phi^{I_h})}) (\beta d(s^{I_h},\theta^{I_h}, \phi^{I_h}) )^n=1$,  we have $\min_{m\in \{0,\cdots, M\}} p_{UN}(s^{I_h},a_m)=\underline{p}(s^{I_h})$ for all $s^{I_h}\in \mathcal{S}$. 
\end{proof}
\begin{remark}[\textbf{Fundamental Limit of ADL}]
Lemma \ref{lemma:SLwithAM} characterizes that the minimum ADL under all AM strategies $a_m\in \mathcal{A}$ is $\underline{p}(s^{I_h})$. The value of $\underline{p}(s^{I_h})$ depends on the operator's capacity gap  ${p}_{CG}(s^{I_h},\theta_{FE}, \phi^{I_h})$ and the frequency of feint and real attacks with different targets, i.e., $\Pr({\theta}^{I_h}, \phi^{I_h}|s^{I_h}), \forall \theta^{I_h}\in \Theta, \phi^{I_h}\in \Phi$.  
\end{remark}

Denote the expected reward of making a complete alert response (i.e., the rewards to dismiss feints and escalate real attacks) as 
\begin{equation*}
\begin{split}
          & \lambda(s^{I_h},m,\phi^{I_h}):= \sum_{i\in \{FE,RE\}} \bar{c}(w_{i},s^{I_h}) \cdot \Pr(\theta_{i}, \phi^{I_h}|s^{I_h})  \\
      & \cdot p_{SP}^h(s^{I_h},\theta_{i},\phi^{I_h}) \cdot [\sum_{n=0}^{m} \frac{1}{n!} e^{-\beta d(s^{I_h},\theta_{i}, \phi^{I_h}) } (\beta d(s^{I_h},\theta_{i}, \phi^{I_h}) )^n].  
\end{split}
\end{equation*}

Combining \eqref{eq:SimpliedConCost} and \eqref{eq:severityLevelFinalform}, we can rewrite ECoC as a combination of the following three terms in \eqref{eq:risk and severitylevel relation}. 
\begin{equation}
    \label{eq:risk and severitylevel relation}
    \begin{split}
       & c(s^{I_h},a_m)  = p_{UN}(s^{I_h},a_m) \bar{c}(w_{UN},s^{I_h})\\
      & \quad\quad   +m\bar{c}(w_{NI},s^{I_h})  + \sum_{\phi^{I_h}\in \Phi}  \lambda(s^{I_h},m,\phi^{I_h}). 
    \end{split}
\end{equation}
Based on Lemma \ref{lemma:SLwithAM}, the first term $p_{UN}(s^{I_h},a_m) \bar{c}(w_{UN},s^{I_h})$ and the third term $\sum_{\phi^{I_h}\in \Phi}  \lambda(s^{I_h},m,\phi^{I_h})$ decrease in $m$, while the second term $m\bar{c}(w_{NI},s^{I_h})$ in \eqref{eq:risk and severitylevel relation} increases in  $m$ linearly at the rate of $\bar{c}(w_{NI},s^{I_h})$. 
The tradeoff among the three terms is summarized below.

\begin{remark}[\textbf{Tradeoff among ADL, Reward of Alert Attention, and Impact for Alert Inattention}]
\label{remark:tradeoff of three}
Based on Lemma \ref{lemma:SLwithAM} and \eqref{eq:risk and severitylevel relation}, increasing $m$ reduces the ADL and achieves a higher reward of completing the alert response. 
However, the increase of $m$ also linearly increases the impact for alert inattention represented by $m \bar{c}(w_{NI},s^{I_h})$, the cost of uninspected alerts. Thus, we need to strike a balance among these terms to reduce the IDoS risk. 
\end{remark}

Define $\lambda_{min}(s^{I_h},\phi^{I_h}):= \sum_{i\in \{FE,RE\}} \bar{c}(w_{i},s^{I_h}) \allowbreak \Pr(\theta_{i}, \phi^{I_h}|s^{I_h}) \allowbreak p_{SP}^h(s^{I_h},\theta_{i},\phi^{I_h})  $,   
$\lambda^{\epsilon_0}_{max}(s^{I_h},\phi^{I_h}):= (1-\epsilon_0) \allowbreak \lambda_{min}(s^{I_h},\phi^{I_h})$, 
$c_{min}(s^{I_h}):=\sum_{\phi^{I_h}\in \Phi} \lambda_{min}(s^{I_h},\phi^{I_h})+ \underline{p}(s^{I_h}) \allowbreak \bar{c}(w_{UN},s^{I_h}) \allowbreak +m\bar{c}(w_{NI},s^{I_h}) $,  
and $c_{max}^{\epsilon_0}(s^{I_h}):= \sum_{\phi^{I_h}\in \Phi} \lambda^{\epsilon_0}_{max}(s^{I_h},\phi^{I_h}) + [\underline{p}(s^{I_h}) +  \epsilon_0(1-\underline{p}(s^{I_h})) ]\allowbreak \bar{c}(w_{UN},s^{I_h}) \allowbreak +m\bar{c}(w_{NI},s^{I_h})$. 

\begin{proposition}
\label{prop:ARiskwithAM}
Consider the scenario where Condition \ref{condition:poissonProcess} holds and $M>\underline{m}(s^{I_h})$. 
For any $\epsilon_0\in (0,1]$ and $s^{I_h}\in \mathcal{S}$, there exists $\underline{m}(s^{I_h})\in \mathbb{Z}^+$ such that 
$c(s^{I_h},a_m)\in [c_{min}(s^{I_h}), c^{\epsilon_0}_{max}(s^{I_h})], \forall a_m\in \mathcal{A}$, when $m\geq \underline{m}(s^{I_h})$. 
Moreover, the lower bound $c_{min}(s^{I_h})$ and the upper bound $c^{\epsilon_0}_{max}(s^{I_h})$ increase in $m$ linearly at the same rate $\bar{c}(w_{NI},s^{I_h})$. 
\end{proposition}

\begin{proof}
For any $\epsilon_0\in (0,1]$, there exists $\underline{m}(s^{I_h})\in \mathbb{Z}^+$ such that $\sum_{n=0}^{m} \frac{1}{n!} e^{-\beta d(s^{I_h},\theta^{I_h}, \phi^{I_h}) } (\beta d(s^{I_h},\theta^{I_h}, \phi^{I_h}) )^n \in [1-\epsilon_0,1]$ when $m\geq \underline{m}(s^{I_h})$. 
Based on Lemma \ref{lemma:SLwithAM}, if $m>\underline{m}(s^{I_h})$, then ${p}_{UN}(s^{I_h},a_m)\in [\underline{p}(s^{I_h}), \underline{p}(s^{I_h})+\epsilon_0(1-\underline{p}(s^{I_h}))]$. 
Plugging it into \eqref{eq:risk and severitylevel relation}, we obtain the results. 
\end{proof}

Let $\sigma^{\underline{m}}\in \Sigma$ denote the AM strategy that chooses to de-emphasize the next $m\geq \underline{m}(s^{I_h})$ alerts for all category label $s^{I_h}\in \mathcal{S}$. 
The monotonicity of the Bellman operator \cite{bertsekas1996neuro} leads to the following corollary. 
\begin{corollary}
\label{corollary:risk}
Consider the scenario where Condition \ref{condition:poissonProcess} holds and $M>\underline{m}(s^{I_h})$. 
For any $\epsilon_0\in (0,1]$ and $s^{I_h}\in \mathcal{S}$, the upper and lower bounds of the risk $u(s^{I_h},\sigma^{\underline{m}})$ increase in $m$ linearly at the same rate of $\bar{c}(w_{NI},s^{I_h})$. 
\end{corollary}

\begin{remark}[\textbf{Fundamental Limit of ECoC and Risk}]
Proposition \ref{prop:ARiskwithAM} and Corollary \ref{corollary:risk} show that the maximum length of the de-emphasized alerts for any $s^{I_h}\in \mathcal{S}$ should not exceed $\underline{m}(s^{hm})$ to reduce the ECoC and the risk of IDoS attacks. 
\end{remark}

\section{Case Study}
\label{sec:case study}
The following section presents case studies to demonstrate the impact of IDoS attacks on human operators' alert inspections and alert responses, and further illustrate the effectiveness of RADAMS. 
Throughout the section, we adopt the attention model in Section \ref{sec:human model}. 

\subsection{Experiment Setup}
\label{subsec:Experiment Setup}

We consider an IDoS attack targeting either the Programmable Logic Controllers (PLCs) in the physical layer or the data centers in the cyber layer of an ICS. 
We \textcolor{black}{denote} these two targets as $\phi_{P}$ and $\phi_{C}$, respectively. 
They constitute the binary set of attack targets $\Phi=\{\phi_{P},\phi_{C}\}$ defined in Section \ref{subsec:feint and real}. 
The SOC of the ICS is in charge of monitoring, inspecting, and responding to both the cyber and the physical alerts. 
We consider two system-level metrics defined in Section \ref{subsec:alert triage}, the source  $\mathcal{S}_{SO}=\{s_{SO,P},s_{SO,C}\}$ and the criticality  $\mathcal{S}_{CR}=\{s_{CR,L},s_{CR,H}\}$, i.e.,  $\mathcal{S}=\mathcal{S}_{SO}\times \mathcal{S}_{CR}$. 
Let $s_{SO,P}$ and $s_{SO,C}$ represent the source of physical and cyber layers, respectively. We assume that the alert triage process can accurately identify the source of attacks, i.e., 
$\Pr(s_{SO,i}|\phi_j)=\mathbf{1}_{\{ i=j  \}}, \forall i,j\in \{P,C\}$. 
Let $s_{CR,L}$ and $s_{CR,H}$ represent low and high  criticality, respectively. 
We assume that the triage process cannot accurately identify feints as low criticality and real attacks as high criticality. 
The revelation kernel is separable and takes the form of  $o(s_{SO},s_{CR}|\theta_i,\phi_j)=\Pr(s_{SO}|\phi_j)\cdot \Pr(s_{CR}|\theta_i),s_{SO}\in \mathcal{S}_{SO},s_{CR}\in \mathcal{S}_{CR}, i\in \{FE,RE\}, j\in\{P,C\}$. 
We choose the values of $o$ so that the attack is more likely to be feint (resp. real) when the criticality level is low (resp. high).

The inter-arrival time at attack stage $k\in \mathbb{Z}^{0+}$ follows an exponential distribution with rate $\beta(\theta^k,\theta^{k+1})$ parameterized by the attack's type $\theta^k,\theta^{k+1}$. 
Thus, the average inter-arrival time $\mu(\theta^k,\theta^{k+1}):=1/\beta(\theta^k,\theta^{k+1})$ also depends on the attack's type at the current and the next attack stages as shown in Table \ref{tab:casevalue average inter-arrival time}. 
We choose the benchmark values based on the literature (e.g., \cite{TIFS8573836,TIFS8470145} and the references within) and attacks can change these values in different IDoS attacks. 

\begin{table}[h]
\centering
\caption{
Benchmark values of the average inter-arrival time $\mu(\theta^k,\theta^{k+1})=1/\beta(\theta^k,\theta^{k+1}), \forall \theta^k,\theta^{k+1}\in \Theta$. 
}
\label{tab:casevalue average inter-arrival time}
\begin{tabular}{|c|c|}
\hline
Average inter-arrival  time from feints to real attacks  & $6$s  \\ \hline
Average inter-arrival time from real attacks to feints  & $10$s \\ \hline
Average inter-arrival time between feints & $15$s \\ \hline
Average inter-arrival time between real attacks & $8$s \\ \hline
\end{tabular}
\end{table}

The average inspection time $\bar{d}$ in Section \ref{subsec:Attentional Factors} depends on the criticality $s^k_{CR}$ and attack's type $\theta^k$ at attack stage $k\in\mathbb{Z}^{0+}$, as shown in Table \ref{tab:casevalue Average INspection time}.  We choose the benchmark values of $\bar{d}(s^k_{CR},\theta^k)$ based on \cite{TIFS8573836}, and these values can change for different human operators and IDoS attacks. 
We add a random noise uniformly distributed in $[-5,5]$ to the average inspection time to simulate the AITN. 

\begin{table}[h]
\centering
\caption{
Benchmark values of the average inspection time  $\bar{d}(s^k_{CR},\theta^k), \forall\theta^k\in \Theta, s^k_{CR}\in \mathcal{S}_{CR}$.  
}
\label{tab:casevalue Average INspection time}
\begin{tabular}{|c|c|}
\hline
Average time to inspect feints of low criticality  & $6$s  \\ \hline
Average time to inspect feints of high criticality  & $8$s \\ \hline
Average time to inspect real attacks of low criticality & $15$s \\ \hline
Average time to inspect real attacks of high criticality & $20$s \\ \hline
\end{tabular}
\end{table}

The stage cost $\bar{c}(w^{k},s^{k}_{SO})$ at attack stage $k\in \mathbb{Z}^{0+}$ in Section \ref{subsec:Stage Cost and Expected Cumulative Cost} depends on the alert response $w^k\in\mathcal{W}$ and the source $s_{SO}^k\in \mathcal{S}_{SO}$.  
We determine the benchmark values of  $\bar{c}(w^{k},s^{k}_{SO})$ per alert in Table \ref{tab:casevalue Costs} based on the salary of the SOC analysts and the estimated loss of the associated attacks.   

\begin{table}[h]
\centering
\caption{
The benchmark values of the stage cost $\bar{c}(w^{k},s^{k}_{SO}), \forall w^k\in\mathcal{W}, s_{SO}^k\in \mathcal{S}_{SO}$. 
}
\label{tab:casevalue Costs}
\begin{tabular}{|c|c|}
\hline
Reward of dismissing feints $w_{FE}$      & \$$80$   \\ \hline
Reward of identifying real attacks $w_{RE}$ in physical layer        & \$$500$     \\ \hline
Reward of identifying real attacks $w_{RE}$ in cyber layer        & \$$100$     \\ \hline
Cost of incomplete alert response $w_{UN}$ or $w_{NI}$         & \$$300$    \\ \hline
\end{tabular}
\end{table}

\subsection{Analysis of Numerical Results}

We plot the dynamics of the operator's alert responses in Fig. \ref{fig:dynamicDecision} under the benchmark experiment setup in Section \ref{subsec:Experiment Setup}.  
We use green, purple, orange, and yellow to represent $w_{UN}$, $w_{NI}$, $w_{FE}$, and $w_{RE}$, respectively. 
The heights of squares are also used to distinguish the four categories. 

\begin{figure}[h]
\centering
 \includegraphics[width=1 \linewidth]{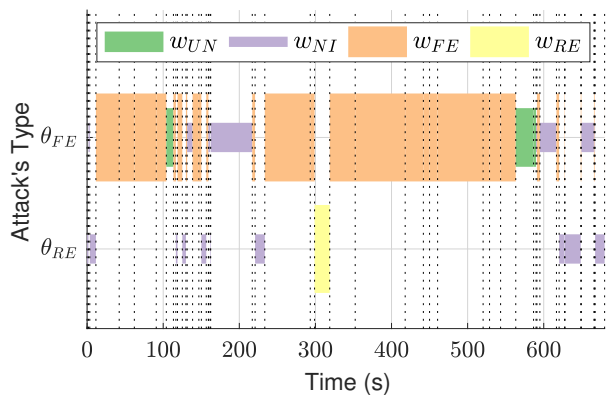}
  \caption{
Alert response $w^k\in \mathcal{W}$ for the $k$-th attack whose type is shown in the $y$-axis. 
The $k$-th vertical dash line represents the $k$-th alert's arrival time $t^k$. 
  }
\label{fig:dynamicDecision}
\end{figure}

\subsubsection{Adaptive Learning during the Real-Time Monitoring and Inspection}
Based on Algorithm \ref{algorithm:human simulation}, we illustrate the learning process of the estimated ECuC $Q^h(s^{I_h},a^h)$ for all $s^{I_h}\in \mathcal{S}$ and $a^h\in \mathcal{A}$ at each inspection stage $h\in \mathbb{Z}^{0+}$ in Fig. \ref{fig:Qlearningtotal}. 
We choose $\alpha^h(s^{I_h},a^h)=\frac{k_c}{k_{TI}(s^{I_h})-1+k_c}$ as the learning rate, where $k_c\in (0,\infty)$ is a constant parameter and $k_{TI}(s^{I_h})\in \mathbb{Z}^{0+}$ is the number of visits to $s^{I_h}\in\mathcal{S}$ up to stage $h\in  \mathbb{Z}^{0+}$. 
Here, the AM action $a^h$ is implemented randomly at each inspection stage $h$, i.e., $\epsilon=1$. Thus, all four AM actions ($M=3$) are explored equally on average for each $s^{I_h}\in \mathcal{S}$ as shown in Fig. \ref{fig:Qlearningtotal}. 
Since the number of visits to different category labels depends on the transition probability $\kappa_{AT}$, the learning stages for four category labels are of different lengths. 

We denote category labels $(s_{SO,P},s_{CR,L})$, $(s_{SO,P},s_{CR,H})$, $(s_{SO,C},s_{CR,L})$, and $(s_{SO,C},s_{CR,H})$ in blue, red, green, and black, respectively. 
To distinguish four AM actions, a deeper color represents a larger $m\in \{0,1,2,3\}$ for each category label $s_{SO,i},s_{CR,j}, i\in \{P,C\}, j\in \{H,L\}$. The inset black box magnifies the selected area. 
The optimal strategy $\sigma^*\in \Sigma$ is to take $a_3$ for all category labels. 
The risk $v^*(s^{I_h})=u(s^{I_h},\sigma^*)$ under the optimal strategy has the approximated values of $\$ 1153$, $\$ 1221$, $\$ 1154$, and $\$ 1358$ for the above category labels  in blue, red, green, and black, respectively. 
Based on Algorithm \ref{algorithm:human simulation}, we also simulate the operator's real-time monitoring and inspection under IDoS attacks when AM strategy is not applied. 
The risks $v^0(s^{I_h}):=u(s^{I_h}, \sigma^0)$ under the default AM strategy $\sigma^0\in\Sigma$ have the approximated values of  $\$1377$,  $\$1527$, $\$1378$,  and $\$ 1620$ for the category label $(s_{SO,P},s_{CR,L})$, $(s_{SO,P},s_{CR,H})$, $(s_{SO,C},s_{CR,L})$, and $(s_{SO,C},s_{CR,H})$, respectively. 
These results illustrate that the optimal AM strategy $\sigma^*\in \Sigma$ can significantly reduce the risk under IDoS attacks for all category labels and the reduction percentage can be as high as $20\%$.

\begin{figure}[h]
\centering
 \includegraphics[width=1\linewidth]{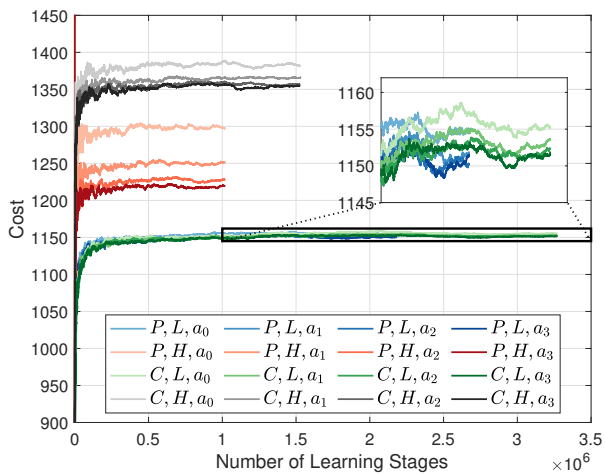}
  \caption{
The convergence of the estimated ECuC $Q^h(s^{I_h},a^h)$ vs. the number of inspection stages. 
  }
\label{fig:Qlearningtotal}
\end{figure}


We further investigate the IDoS risk under the optimal AM strategy $\sigma^*$ as follows. 
As illustrated in Fig. \ref{fig:Qlearningtotal}, when the criticality level is high (i.e., the attack is more likely to be real), the attacks targeting cyber layers (denoted in black) result in a higher risk than the one targeting physical layers (denoted in red).  
This asymmetry results from the different rewards of identifying real attacks in physical or cyber layers denoted in Table \ref{tab:casevalue Costs}. 
Since dismissing feints brings the same reward in physical and cyber layers, the attacks targeting physical or cyber layers result in similar IDoS risks when the criticality level is low. 
Within physical or cyber layers, high-criticality alerts result in a higher risk than low-criticality alerts do. 


The value of $Q^{h}(s^{I_h},a_m), m\in \{0,1,2\}$, represents the risk when RADAMS deviates to sub-optimal AM action $a_m$ for a single category label $s^{I_h}\in\mathcal{S}$. 
As illustrated by the red and black lines in Fig. \ref{fig:Qlearningtotal}, this single deviation can increase the risk under alerts of high criticality. 
However, it hardly increases the risk under alerts of low criticality as illustrated by the green and blue lines in the inset black box of Fig. \ref{fig:Qlearningtotal}. 
These results illustrate that we can deviate from the optimal AM strategy to sub-optimal ones for some category labels with approximately equivalent risk, which we refer to as the \textit{attentional risk equivalency} in Remark \ref{remark:deviation Equivalency}. 

\begin{remark}[\textbf{Attentional Risk Equivalency}] 
\label{remark:deviation Equivalency}
The above results illustrate that we can contain the IDoS risk by selecting proper sub-optimal strategies. 
If applying the optimal AM strategy  $\sigma^*$ is costly, then RADAMS can choose not to apply AM strategy for $(s_{SO,C},s_{CR,L})$ or $(s_{SO,P},s_{CR,L})$ without significantly increasing the IDoS risks. 
\end{remark} 

\subsubsection{Optimal AM Strategy and Resilience Margin under Different Stage Costs}

We define \textit{resilience margin} as the difference of the risks under the optimal and the default AM strategies. 
We investigate how the cost of incomplete alert response in Table \ref{tab:casevalue Costs} affects the optimal AM strategy and the resilience margin in Fig. \ref{fig:costchange}. 


\begin{figure}[h]
\centering
 \includegraphics[width=1\linewidth]{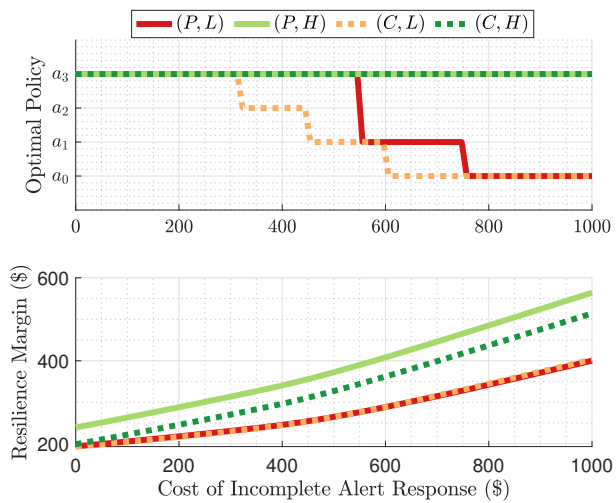}
  \caption{
The optimal AM strategy and the risk vs. the cost of an incomplete alert response under category label ($s_{SO,P},s_{CR,L}$), ($s_{SO,P},s_{CR,H}$), ($s_{SO,C},s_{CR,L}$), and ($s_{SO,C},s_{CR,H}$) in solid red, solid green, dashed yellow, and dashed green, respectively. 
}
\label{fig:costchange}
\end{figure}

As shown in the upper figure, the optimal strategy remains to choose AM action $a_3$ when the alert is of high criticality. 
When the alert is of low criticality, then as the cost increases, the optimal AM strategy changes sequentially from $a_3$, $a_2$, and $a_1$ to $a_0$; i.e., RADAMS gradually decreases $m\in \{0,1,2,3\}$, the number of de-emphasized alerts. 
As shown in the lower figure, the resilience margin increases monotonously with the cost. 
The optimal strategy for alerts of high criticality yields a larger resilience margin than the one for low criticality. 
\begin{remark}[\textbf{Tradeoff of Monitoring and Inspection}]
The results show that the optimal strategy strikes a balance between real-time monitoring a large number of alerts and inspecting selected alerts with high quality. Moreover, the optimal strategy is resilient for a large range of cost values ($[\$ 0,\$  1000]$). 
If the cost is high, and the alert is of low (resp. high) criticality, then the optimal strategy encourages monitoring (resp. inspecting) by choosing a small (resp. large) $m$. 
However, when the cost of an incomplete alert response is relatively low, the optimal strategy is $a_4$ for all alerts because the high-quality inspection outweighs the high-quantity monitoring. 
\end{remark}

\subsubsection{Arrival Frequency of IDoS Attacks}
\label{subsec:Arrival Frequency of IDoS Attacks}

As stated in Section \ref{subsec:feint and real}, feint attacks with the goal of triggering alerts require fewer resources to craft. 
Thus, we let $\hat{c}_{RE}=\$ 0.04$ and $\hat{c}_{FE}\in (0,\hat{c}_{RE})$ denote the cost to generate a real attack and a feint, respectively.  
With $\hat{c}_{RE}$ and $\hat{c}_{FE}$, we can compute the attack cost of feint and real attacks per work shift of $24$ hours. 
Let $\rho$ be the scaling factor for the arrival frequency, and in Section \ref{subsec:Arrival Frequency of IDoS Attacks}, the average inter-arrival time is $\hat{\mu}(\theta^k, \theta^{k+1})= \rho {\mu}(\theta^k, \theta^{k+1}), \forall \theta^k, \theta^{k+1}\in \Theta$. 
We investigate how the scale factor $\rho\in (0,2.5]$ affects the IDoS risk and the attack cost in Fig. \ref{fig:freq_risk}. 
As $\rho$ decreases, the attacker generates feint and real attacks at a higher frequency. 
Then, the risks under both the optimal and the default strategies increase. However, the optimal AM strategy can reduce the increase rate for a large range of $\rho\in [0.5,2]$. 

\begin{figure}[h]
\centering
 \includegraphics[width=1 \linewidth]{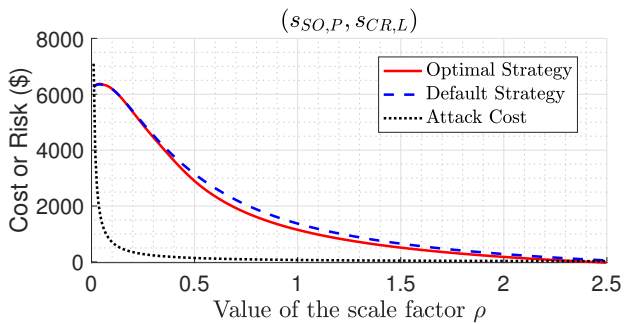}
  \caption{
IDoS risk vs. $\rho$ under the optimal and the default AM strategies in solid red and dashed blue, respectively. 
The black line represents the attack cost per work shift of $24$ hours. 
  }
\label{fig:freq_risk}
\end{figure}  

\begin{remark}[\textbf{Attacker's Dilemma}]
From the attacker's perspective, although increasing the attack frequency can induce a high risk to the organization, and the attacker can gain from it, the frequency increase also increases the attack cost exponentially, as shown by the dotted black line in Fig. \ref{fig:freq_risk}. 
Thus, the attacker has to strike a balance between the attack cost and the attack gain (represented by the IDoS risk). 
Moreover, attackers with a limited budget are not capable to choose small values of $\rho$ (i.e., high attack frequencies). 
\end{remark}



\subsubsection{Percentage of Feint and Real Attacks}

Consider the case where $\kappa_{AT}$ independently generates feints and real attacks with probability $\eta_{FE}$ and $\eta_{RE}=1-\eta_{FE}$, respectively. 
We consider the case where the attacker has a limited budge $\hat{c}_{max}=\$ 270$ per work shift (i.e., $86400 s$) and generates feint and real attacks at the same rate $\hat{\beta}$, i.e., $\beta(\theta^k, \theta^{k+1})=\hat{\beta}, \forall \theta^k,\theta^{k+1}\in \Theta$. 
Consider the attack cost in Section \ref{subsec:Arrival Frequency of IDoS Attacks}, the attacker has the following budget constraint, i.e., 
\begin{equation}
\label{eq:attacker budget constraint}
    86400 \cdot \hat{\beta} \cdot  (\eta_{FE} \hat{c}_{FE}+\eta_{RE} \hat{c}_{RE}) \leq \hat{c}_{max}. 
\end{equation}
The budget constraint results in the following tradeoff.  
If the attacker chooses to increase the probability of real attack $\eta_{RE}$, then he has to reduce the arrival frequency $\hat{\beta}$ of feint and real attacks. 
We investigate how the probability of feints affects the IDoS risk in Fig. \ref{fig:prob of feints} under the optimal and the default AM strategies in red and blue, respectively. 
The feints are of low and high costs in Fig. \ref{fig:feintuFEeqRE_cef1} and \ref{fig:feintuFEeqRE_cef5}, respectively. 

\begin{figure}[h]
    \centering 
    \begin{subfigure}{0.24\textwidth}
 \includegraphics[width=1\linewidth]{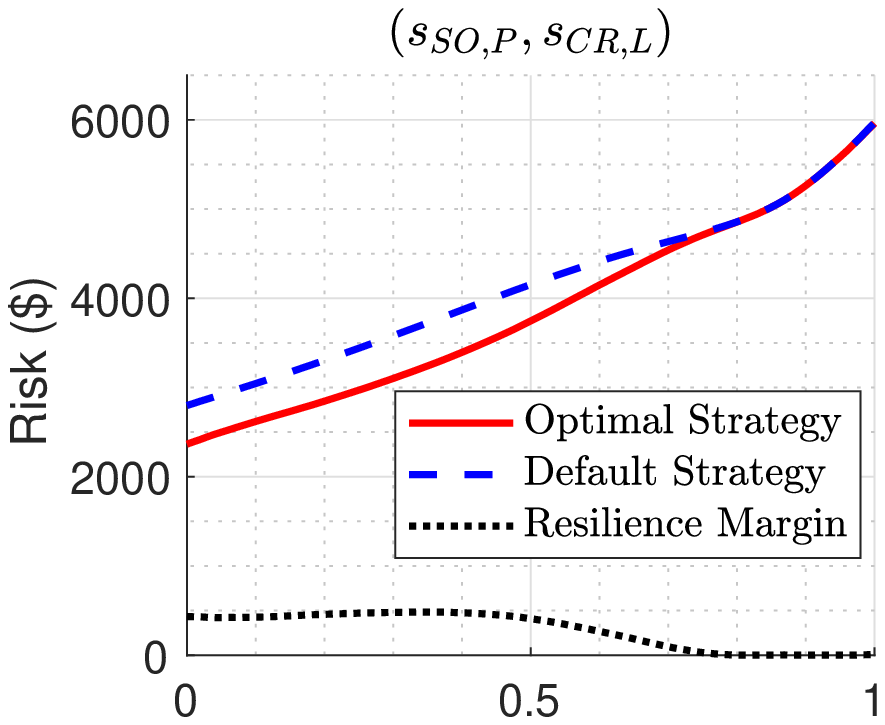} 
  \caption{
Low-cost feints $\hat{c}_{FE}=1$. 
  }
\label{fig:feintuFEeqRE_cef1}
\end{subfigure}\hfil 
    \begin{subfigure}{0.24\textwidth}
 \includegraphics[width=1\linewidth]{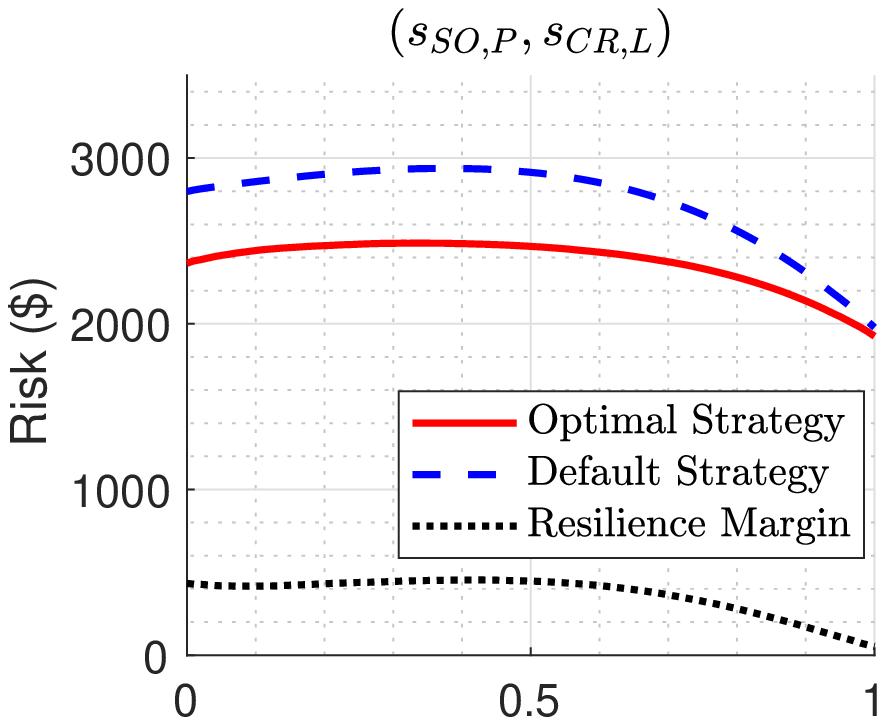} 
  \caption{
High-cost feints $\hat{c}_{FE}=5$. 
  }
\label{fig:feintuFEeqRE_cef5}
\end{subfigure}\hfil 
\caption{
IDoS risk vs. $\eta_{FE}\in [0,1]$ under the optimal and the default AM strategies in red and blue, respectively. 
The black line represents the resilience margin. 
\label{fig:prob of feints}
}
\end{figure}

As shown in Fig. \ref{fig:feintuFEeqRE_cef1}, when the feints are of low cost, i.e., $\hat{c}_{FE}=\hat{c}_{RE}/10$, generating feints with a higher probability monotonously increases the IDoS risks for both AM strategies. 
When the probability of feints is higher than $80\%$, the resilience margin is zero; i.e., the optimal and the default AM strategies both induce high risks. 
However, as the probability of feint decreases, the resilience margin increases to around $\$500$; i.e., the default strategy can moderately reduce the risk, but the optimal strategy can excessively reduce the risk. 

\begin{remark}[\textbf{Half-Truth Attack for High-Cost Feints}]
As shown in Fig. \ref{fig:feintuFEeqRE_cef5}, when the feints are of high cost, i.e., $\hat{c}_{FE}=\hat{c}_{RE}/2$, then the optimal attack strategy is to deceive with \textit{half-truth}, i.e., generating feint and real attacks with approximately equal probability to induce the maximum IDoS risk. 
As the probability of feints decreases from $\eta_{FE}=1$, the risk increases significantly under the default AM strategy but moderately under the optimal one. 
\end{remark}

The figures in Fig. \ref{fig:prob of feints} show that the optimal attack strategy under the budget constraint \eqref{eq:attacker budget constraint} needs to adapt to the cost of feint generation. 
Regardless of the attack strategy, the optimal AM strategy can reduce the risk and achieve a positive resilient margin for all category labels $(s_{SO,i},s_{CR,j}), i\in \{P,C\},j\in \{L,H\}$. 
Moreover, higher feint generation cost reduces the arrival frequency of IDoS attacks due to \eqref{eq:attacker budget constraint}. 
Thus, comparing to Fig. \ref{fig:feintuFEeqRE_cef1}, the risk in Fig. \ref{fig:feintuFEeqRE_cef5} is lower for the same $\eta_{FE}$ under the optimal or the default AM strategies, especially when $\eta_{FE}$ is close to $1$.

\subsubsection{The Operator's Attention Capacity}

We consider the following attention function $f_{LOE}\circ f_{SL}$ with a constant attention threshold, i.e.,  $\bar{n}(y_{EL},s^k)=\bar{n}_0, \forall y_{EL},s^k\in\mathcal{S}$.
Consider the following trapezoid attention function. 
If $n^t\leq \bar{n}_0$, the LOE $\omega^t=1$; i.e., the operator can retain the high LOE when the number of distractions is less than the attention threshold $\bar{n}_0$. If $n^t > \bar{n}_0$, the LOE $\omega^t$ gradually decreases as $n^t$ increases. 
Then, a larger value of $\bar{n}_0$ indicates a high attention capacity. We investigate how the value of $\bar{n}_0$ affects the risk in Fig. \ref{fig:nu}. 

\begin{figure}[h]
\centering
 \includegraphics[width=1\linewidth]{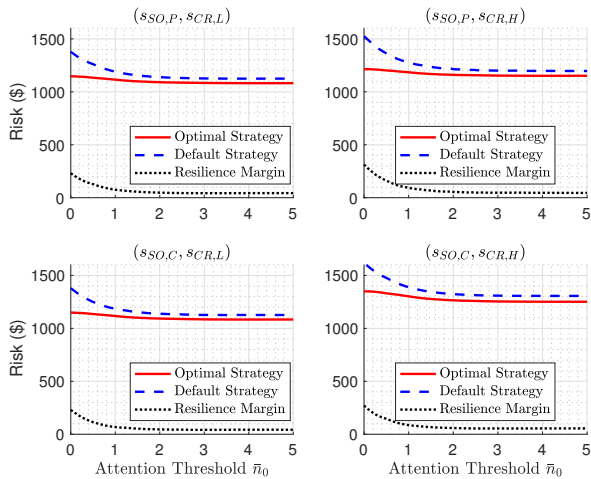}
  \caption{
  Risk vs. attention threshold under the optimal and the default AM strategies in red and blue, respectively. The black dotted line represents the resilience margin. 
  }
\label{fig:nu}
\end{figure}
 As the operator's attention capacity increases, the risks under the optimal and the default AM strategies decrease for all category labels. 
 The resilience margin decreases from around $\$ 200$ to $\$ 50$ as $\bar{n}_0$ increases from $0$ to $2$ and then maintains the value of around $\$ 50$. 
 Thus, the optimal strategy suits operators with a large range of attention capacity, especially for the ones with limited attention capacity.

\section{Conclusion}
\label{sec:conclusion}

Attentional human vulnerabilities exploited by attackers lead to a new class of proactive attacks called the Informational Denial-of-Service (IDoS) attacks. 
IDoS attacks generate a large number of feint attacks on purpose to deplete the limited human attention resources and exacerbate the alert fatigue problem. 
In this work, we have formally defined IDoS attacks as a sequence of feint and real attacks of heterogeneous targets, which can be characterized by the Markov renewal process. 
We have abstracted the alert generation and \textcolor{black}{technical-level} triage processes as a revelation probability to establish a stochastic relationship between the IDoS attack's hidden types and targets and the associated alert's observable category labels. 
We have explicitly incorporated human factors (e.g., levels of expertise, stress, and efficiency) and empirical results (e.g., the Yerkes–Dodson law and the sunk cost fallacy) to model the operators' attention dynamics and the processes of alert monitoring, inspection, and response in real time. 
Based on the system-scientific human attention and alert response model, we have developed a Resilient and Adaptive Data-driven alert and Attention Management Strategy (RADAMS) to assist human operators in combating IDoS attacks. 
We have proposed a Reinforcement Learning (RL)-based algorithm to obtain the optimal assistive strategy according to the costs of the operator's alert responses in real time. 

Through theoretical analysis, we have observed the \textit{Product  Principle of  Attention}  (PPoA), the fundamental limits of Attentional Deficiency Level (ADL) and risk, and tradeoff among the ADL, the reward of alert attention, and the impact of alert inattention. 
Through the experimental results, we have corroborated the \textit{effectiveness}, \textit{adaptiveness}, \textit{robustness}, and \textit{resilience} of the proposed assistive strategies as follows. 
First, the optimal AM strategy outperforms the default strategy and can effectively reduce the IDoS risk by as much as $20\%$. 
Second, the strategy adapts to different category labels to strike a balance of monitoring and inspections. 
Third, the optimal AM strategy is robust to deviations. 
We can apply sub-optimal strategies at some category labels without significantly increasing the IDoS risk. 
Finally, the optimal AM strategy is resilient to a large variations of costs, attack frequencies, and human attention capacities. 

\textcolor{black}{The current work uses Industrial Control Systems (ICS) as a quintessential example to illustrate the IDoS attacks and the associated human-aware alert and attention management strategies. RADAMS can also be applied to broad types of scenarios (e.g., healthcare, public transport control, and weather warning) that require human operators of limited attention resources to monitor and manage massive alerts in real time with a high level of situational awareness. 
RADAMS adopts the ``less is more'' principle by restricting the amount of information processed by the human operators to be within their attention capacities. 
Such principle is transferable to other assailable cognitive resources of human operators, including memory, reasoning, and learning capacity. 
The future work would incorporate more generalized models (e.g., the spatio-temporal self-excited process) to capture the history-dependent temporal arrival of IDoS attacks, the spatial location of the alerts, their impacts on human attention, and the associated human-assistive security technologies.}


\bibliographystyle{IEEEtran}
\bibliography{IDOSjournal}

\end{document}